\documentclass[letterpaper]{scrartcl}
\usepackage[total={16.2cm, 23.5cm}]{geometry}
\usepackage{microtype}
\usepackage[T1]{fontenc}
\usepackage{amsmath,amsthm,amssymb, mathtools}
\usepackage{tikz}
\usetikzlibrary{patterns,decorations.pathreplacing, arrows.meta}
\usepackage{hyperref}
\usepackage{soul}
\hypersetup{
pdfencoding=auto, psdextra,
colorlinks=true,citecolor=green!30!black,linkcolor=red!40!black,urlcolor=blue!60!black
}
\usepackage[sort&compress,nameinlink]{cleveref}
\usepackage{nicefrac}
\usepackage[numbers,sort]{natbib}
\usepackage{comment}
\usepackage{paralist}
\usepackage[shortlabels]{enumitem}
\usepackage{xcolor}
\usepackage{colortbl}
\usepackage{subcaption}
\usepackage{booktabs}
\usepackage{listings}
\usepackage{complexity}
\usepackage[linesnumbered,ruled,vlined]{algorithm2e}
\usepackage[linecolor=green!70!black, backgroundcolor=green!10, bordercolor=black, textsize=small]{todonotes}
\definecolor{winered}{rgb}{0.5,0.1,0.1}

\renewcommand*{\le}{\leqslant}
\renewcommand*{\leq}{\leqslant}
\renewcommand*{\ge}{\geqslant}
\renewcommand*{\geq}{\geqslant}
\renewcommand{\epsilon}{\varepsilon}

\crefname{table}{Table}{Tables}
\Crefname{table}{Table}{Tables}
\crefname{figure}{Figure}{Figures}
\crefname{theorem}{Theorem}{Theorems}
\crefname{definition}{Definition}{Definitions}
\crefname{corollary}{Corollary}{Corollaries}
\crefname{observation}{Observation}{Observations}
\crefname{question}{Question}{Question}
\crefname{lemma}{Lemma}{Lemmas}
\crefname{example}{Example}{Examples}
\crefname{reduction}{Reduction}{Reductions}
\crefname{construction}{Construction}{Constructions}
\crefname{subsection}{Section}{Sections}
\crefname{section}{Section}{Sections}
\crefname{proposition}{Proposition}{Propositions}
\crefname{algorithm}{Algorithm}{Algorithms}
\crefname{algocf}{Algorithm}{Algorithms}
\Crefname{equation}{Inequality}{Inequalities}
\crefname{lstlisting}{listing}{listings}

\newcommand{\myemph}[1]{{\color{winered}\emph{#1}}}
\newcommand{\naturals}{{{\mathbb{N}}}}
\newcommand{\feasibles}[1][]{{{\mathcal{F}_{#1}}}}
\newcommand{\reals}{{{\mathbb{R}}}}

\newcommand{\weight}{{{\mathfrak{w}}}}
\newcommand{\lquota}[1]{{{q^{\bot}_{#1}}}}
\newcommand{\uquota}[1]{{{q^{\top}_{#1}}}}

\newcommand{\pav}{{{\mathrm{PAV}}}}

\newcommand{\score}{{{\mathrm{score}}}}

\renewcommand{\part}{{{\mathrm{part}}}}
\renewcommand{\H}{{{\mathrm{H}}}}

\newcommand{\calR}{{{\mathcal{R}}}}

\theoremstyle{definition}
\newtheorem{definition}{Definition}
\newtheorem{example}{Example}
\newtheorem{remark}{Remark}
\theoremstyle{plain}
\newtheorem{theorem}{Theorem}
\newtheorem{lemma}[theorem]{Lemma}

\newtheorem{proposition}[theorem]{Proposition}

\makeatletter
\newtheorem*{rep@theorem}{\rep@title}
\newcommand{\newreptheorem}[2]{%
\newenvironment{rep#1}[1]{%
 \def\rep@title{#2 \ref{##1}}%
 \begin{rep@theorem}}%
 {\end{rep@theorem}}}
\makeatother
\newreptheorem{theorem}{Theorem}
\newreptheorem{proposition}{Proposition}
\newreptheorem{lemma}{Lemma}
\newreptheorem{corollary}{Corollary}

\DeclareMathOperator*{\argmax}{argmax}

\begin{document}
\title{A Generalised Theory of Proportionality in Collective Decision Making}
\author{
  Tomáš Masařík\\
  University of Warsaw\\
  \href{mailto:masarik@mimuw.edu.pl}{masarik@mimuw.edu.pl}
  \and
  Grzegorz Pierczyński\\
  University of Warsaw\\
  \href{mailto:g.pierczynski@mimuw.edu.pl}{g.pierczynski@mimuw.edu.pl}
  \and
Piotr Skowron\\
  University of Warsaw\\
  \href{mailto:p.skowron@mimuw.edu.pl}{p.skowron@mimuw.edu.pl}
}
\date{}
	\maketitle

\begin{abstract}
We consider a voting model, where a number of candidates need to be selected subject to certain feasibility constraints. The model generalises committee elections (where there is a single constraint on the number of candidates that need to be selected), various elections with diversity constraints, the model of public decisions (where decisions needs to be taken on a number of independent issues), and the model of collective scheduling. A critical property of voting is that it should be fair---not only to individuals but also to groups of voters with similar opinions on the subject of the vote; in other words, the outcome of an election should proportionally reflect the voters' preferences.
We formulate axioms of proportionality in this general model. Our axioms do not require predefining groups of voters; to the contrary, we ensure that the opinion of every subset of voters whose preferences are cohesive-enough are taken into account to the extent that is proportional to the size of the subset. Our axioms generalise the strongest known satisfiable axioms for the more specific models. We explain how to adapt two prominent committee election rules, Proportional Approval Voting (PAV) and Phragm\'{e}n Sequential Rule, as well as the concept of stable-priceability to our general model. The two rules satisfy our proportionality axioms if and only if the feasibility constraints are matroids.
\end{abstract}

\lstset{
    keywords={input, output, for, while, if, else, return, break},
    comment=[l]{//},
    frame=single,
    mathescape=true,
    float,
    captionpos=b,
    numbers=left,
    breaklines=true,
}

\section{Introduction}

We consider a general voting scenario, where a subset of candidates needs to be selected based on the voters' preferences. The generality of this model comes from the fact that we do not consider specific types of elections, but rather assume we are given \myemph{feasibility constraints} as a part of an election. The constraints encode the type of the election by specifying which subsets of candidates can be elected. For example, if the goal is to select a fixed numer of candidates, say $k$ of them, then the constraints would simply indicate that all $k$-element subsets of the candidates are feasible. Naturally, the model allows to incorporate additional diversity constraints that specify lower and upper bounds on the number of selected candidates from different demographic groups.

Yet, the general feasibility constraints give much more flexibility and allow us to capture considerably more complex scenarios, which at first might seem not to be about selecting subsets of candidates. For example, consider the setting of public decisions where we need to make decisions on a number of independent issues~\cite{conitzer2017fair,fre-kah-pen:variable_committee_elections,sko-gor:proportional_public_decisions,cha-goe-pet:seq-decision-making,lac:perpetual-voting,bulteau2021jr}. For each alternative we can introduce a candidate and the feasibility constraints would indicate that exactly one alternative needs to be selected on each issue. Similarly, consider a model where the voters provide partial orders over the candidates and the goal is to establish a ranking of the candidates~\cite{dwo-kum-nao-siv:c:rank-aggregation,proprank}. This can be also expressed in our model by introducing auxiliary candidates: for each pair of candidates, $c_i$ and $c_j$, an auxiliary candidate $c_{i, j}$ would indicate that $c_i$ is ranked before $c_j$ (either in the resulting ranking or in the voters' ballots). Our model also captures committee elections with negative votes~\cite{bau-den:mav_trichotomous,zhou2019parameterized} and judgement aggregation~\cite{lis-pol:judgment-aggregation,End15JA}---we explain this in more detail in \Cref{sec:model}.

While the aforementioned types of elections might appear very different, certain common high-level principles apply to all of them.
In particular, in most scenarios it is of utmost importance to ensure that the outcomes of elections are fair---not only to individuals but also to groups of voters with similar views. Indeed, fair elections provide equal opportunities for underrepresented groups to engage in the process of decision-making, and lead to more inclusive and accountable decisions. Fairness has also a positive effect on participation and enhances the legitimacy of the elected candidates.\footnote{Fairness is also critically important in elections that do not involve humans. For example, proportional election rules are used for selecting validators in the blockchain protocol~\cite{cevallos2020validator} (proportionality is important to provide resilience against coordinated attacks of malicious users) or for improving the quality of genetic algorithms~\cite{fal-saw-sch-smo:c:multiwinner-genetic-algorithms}.} Accordingly, group-fairness in elections is the central focus of this paper.\footnote{While there is a large literature on individual fairness in social choice, it mainly deals with the problem of fair allocation of private goods~\cite{Thomson16:fair-allocation,BouveretCM16:fair-allocation}. In the context of voting, group-fairness is a more compelling concept. Indeed, the selected candidates are typically not assigned to particular individuals, but they rather represent (and are selected based on votes of) whole groups of voters~\cite{BaYo82a,Puke14a,FSST-trends}.}

There are at least three major ways in which one can implement group-fairness:
\begin{enumerate}
  \item One can additionally collect detailed socio-demographic information about the voters and use this information when computing winners of the election~\cite{pes-shm:algo-fairness}. This approach has severe drawbacks---it elicits sensitive information, and requires an in-depth a priori knowledge on what part of personal data is important for guaranteeing fairness.
    This is \emph{not} the approach that we take in this paper. We will still guarantee, though, that each socio-demographic group has a proportional influence on the outcome, assuming that that the members of such groups vote similarly.

\item One can specify certain concrete diversity constraints~\cite{azi:committees-soft-constraints,conf/aaai/BredereckFILS18,conf/ijcai/CelisHV18}, i.e., constraints that enforce that the outcomes of elections have certain structure, such as gender equality. Our model provides the possibility of adding such constraints. However, this is not sufficient, since the diversity constraints are typically fixed, and do not depend on the voters' preferences.

\item Finally, we can ensure proportional aggregation of the voters' preferences. Here, the main idea is that each subset of voters who support similar candidates should be entitled to influence the decision to the extent proportional to the size of the subset. This approach has been pursued in recent years in the literature on committee elections~\cite{FSST-trends,lac-sko:multiwinner-book}, and in this paper we extend it to more general types of elections.
\end{enumerate}

Our main contribution is conceptual. We define several axioms of proportionality in the general model of elections with feasibility constraints. Our axioms differ in their strength, but they share the same intuitive explanation: if a group of voters has cohesive-enough preferences, then they should have the right to decide about a proportional part of the elected outcome. Our axioms generalize the strongest known properties from the literature on committee elections, namely fully justified representation (FJR)~\cite{pet-pie-sko:c:participatory-budgeting-cardinal}, extended justified representation (EJR)~\cite{justifiedRepresentation}, and proportional justified representation (PJR)~\cite{pjr17}. One of the main results of this paper says that our axioms are always satisfiable---the base axioms are satisfiable for all types of constraints, and the most demanding strengthening of these axioms are satisfiable for matroid constraints.

We further explain how to adapt two prominent committee election rules, Proportional Approval Voting (PAV) and Phragm\'{e}n Sequential Rule, to our general model. We provide a full characterisation explaining that the two rules satisfy the aforementioned proportionality axioms if and only if the feasibility constraints are matroids (\Cref{thm:pav-and_ejr,thm:phragmen-and_pjr}). We also adapt the concept of stable-priceability to the model with general constraints, and we prove that the solutions that are stable-priceable satisfy our strong notions of proportionality. Altogether, our results provide basic tools that allow to guarantee group-fairness in different types of elections and in the presence of different types of constraints.

\subsubsection*{Related Work}

We discuss most of the related work in the relevant parts of the paper. Here we mention yet another particularly pertinent line of research, even though it does not directly relate to our axioms. One of the strongest notion of group-fairness considered in the social-choice literature is the core~\cite{justifiedRepresentation}. Since for some types of elections the core might be empty, certain relaxations of the core are often considered, for example its approximate variants~\cite{FMS18,mav-mun-she:committees-with-constraints,cheng2019group,jiang2019approx,min-yah-wan:core-approx}. Among this literature the recent work of \citet{mav-mun-she:committees-with-constraints} is most closely related to our paper. This work alike considers the model with diversity constraints, yet the considered axioms are very different (we discuss this difference in detail in \Cref{sec:base_ejr_vs_restrained_ejr}).

\section{The Model}\label{sec:model}

For each natural number $t\in \naturals$ we set $[t]=\{1, 2, \ldots, t\}$, and we use the convention that $[0] = \emptyset$.

An election is a quadruple $E=(C, N, \feasibles, \mathcal{A})$, where $C=\{c_1, \ldots, c_m\}$ is a set of $m$ \myemph{candidates}, $N=\{1, 2, \ldots n\}$ is a set of $n$ \myemph{voters}, $\feasibles\subseteq 2^C$ is a nonempty family of \myemph{feasibility sets}, and $\mathcal{A}=(A_1,\ldots,A_n)$ is a collection of \myemph{approval ballots};  $A_i \subseteq C$ for each voter $i \in N$. Intuitively, a voter's approval ballot is a subset of candidates that the voter supports.
Analogously, for a candidate $c \in C$ by $N(c)$ we denote the set of voters that approve $c$, $N(c) = \{i \in N\colon c \in A_i\}$.
For each voter $i$ we define her utility from a subset $W \subseteq C$ as the number of candidates in $W$ that the voter approves, that is $u_i(W) = |A_i \cap W|$.
Intuitively, $u_i(W)$ quantifies the satisfaction that voter $i$ enjoys if the subset $W$ is selected.
In \Cref{sec:geenral-utility-functions} we additionally discuss more expressive types of ballots, and the corresponding more complex utility functions.

We say that a subset of candidates $W\subseteq C$ is \myemph{feasible} if $W\in \feasibles$.
A \myemph{selection rule} is a function $\calR$ that given an election returns a nonempty set of feasible outcomes.\footnote{We are typically interested in a single outcome, yet we allow for ties.}
Without loss of generality, we assume that $\feasibles$ is closed under inclusion, i.e., if $W \in \feasibles$ and $W' \subseteq W$, then $W' \in \feasibles$. Indeed, if $W' \notin \feasibles$, it could be completed to a feasible set, and the voters would enjoy at least as high utility from the completed set as from the original one. Accordingly, when defining feasibility constraints, we often indicate only the maximal sets and implicitly assume that all the subsets are also feasible.



\subsection*{Feasibility Constraints}\label{sec:feasibility_constraints}

Our framework generalizes several important models considered in the literature, in particular:
\begin{description}
\item[Committee elections.] Here, we assume the goal is to select a subset of candidates (called a committee) of a given fixed size $k$. Thus, the feasibility constraints are of the following form:
\begin{align*}
\feasibles = \left\{ W \subseteq C \colon |W| = k\right\} \text{.}
\end{align*}
The model of committee elections has been extensively studied in the literature; we refer to the book by \citet{lac-sko:multiwinner-book} and to the book chapter by \citet{FSST-trends}.
\item[Public Decisions.] Here, we assume that the set of candidates is divided into $z$ disjoint pairs $C = \bigcup_{r \in [z]} C_r$, $|C_r| = 2$ for each $r \in [z]$ and $C_r \cap C_s = \emptyset$ for all $r, s \in [z]$ with $r \neq s$. For each pair we must select a single candidate, thus, the feasibility constraints are given as:
\begin{align*}
\feasibles = \left\{ W \subset C \colon |W \cap C_r| = 1 \text{~for~each~} r \in [z]\right\} \text{.}
\end{align*}
Intuitively, each pair corresponds to an issue on which a binary decision needs to be made; one candidate in the pair corresponds to the ``yes''-decision, and the other one to the ``no''-decision. This model has been studied by \citet{fre-kah-pen:variable_committee_elections}, and \citet{sko-gor:proportional_public_decisions}. One particularly appealing application domain for this model is to support negotiations among groups of entities in order to establish a common policy (e.g., negotiations among political parties that want to form a governing coalition). A variant of this model, where for each issue $r$ more than two alternative options are available---$|C_r| \geq 2$, has been also considered in the literature~\cite{conitzer2017fair,bri-mar-pap-pet:proporitonlaity-on-independent-issues}.
\end{description}

In this paper we additionally introduce an intermediate model that is more specific than the general model with arbitrary feasibility constraints, yet still more expressive than the models of committee elections and public decisions. This model is similar to the one of committee elections with diversity constraints, as studied by \citet{conf/aaai/BredereckFILS18}, \citet{conf/ijcai/CelisHV18}, and \citet{azi:committees-soft-constraints}.
\begin{description}
\item[Committee elections with disjoint attributes.] Here, we assume that the set of candidates is divided into $z$ disjoint groups $C = \bigcup_{r \in [z]} C_r$, $C_r \cap C_s = \emptyset$ for $r, s \in [z]$, $r \neq s$. For each group $r \in [z]$ we are given two numbers: a lower and an upper quota, denoted respectively as $\lquota{r}$ and $\uquota{r}$. The goal is to select $k$ candidate so that the number of candidates selected from each set $C_r$ is between $\lquota{r}$ and $\uquota{r}$.
\begin{align*}
\feasibles = \left\{ W \subseteq C \colon |W| = k \text{~and~} \lquota{r} \leq |W \cap C_r| \leq \uquota{r} \text{~for~each~} r \in [z]\right\} \text{.}
\end{align*}
\end{description}

The feasibility constraints mentioned above are all special cases of a more general class of constraints with a matroid structure~\cite{oxl-matroids}. We prove this in \Cref{sec:matroid}.

\begin{definition}[Matroid constraints] The feasibility constraints are matroid if the following condition, called the \emph{exchange property}, is satisfied:\footnote{Formally, this means that the pair $(C, \feasibles)$ forms a matroid.}
\begin{enumerate}[label=(EP)]
    \item\label{cond:exchange-property} For each $X, Y \in \feasibles$ such that $|X| < |Y|$, there exists $c\in Y\setminus X$ such that $X\cup\{c\}\in \feasibles$.
\end{enumerate}
\end{definition}

Intuitively, in a matroid all the candidates carry the same weight in the constraints. If we can remove some two candidates to make space for some other candidate $c$, then it is sufficient to remove only one of these two candidates. This is formalised in the following lemma.

\begin{lemma}\label{lem:one-swap}
    For all elections with matroid constraints, all feasible sets $W\subseteq C$, $c\notin W$ and $W'\subseteq W$ ($W'\neq \emptyset$) such that $W\setminus W' \cup \{c\}\in \feasibles$ there exists $c'\in W'$ such that $(W\setminus\{c'\}) \cup\{c\}\in \feasibles$.
\end{lemma}
\begin{proof}
Suppose that \ref{cond:exchange-property} is satisfied and the statement is violated. Consider $W$, $W'$ and $c$ witnessing this violation. Assume that the witness is chosen so that $|W'|$ is minimised (yet still $|W'| > 1$, as otherwise the only member of $W'$ would be clearly the required candidate). Now let us define sets $X, Y\in \feasibles$ as follows: $X = W\setminus W' \cup \{c\}$, $Y=W$. From the observation that $|W'| > 1$ we have that $|X| < |Y|$. Then by \ref{cond:exchange-property} we have that there exists $a\in Y \setminus X = W'$ such that $X\cup\{a\}\in \feasibles$. Therefore, after removing $a$ from $W'$ we would obtain a smaller witness of the violation of the lemma statement, a contradiction.
\end{proof}

While a large part of our results concerns matroids, our definitions also apply to computational social choice models that do not have a matroid structure. Below we give a few examples of such models that fit our general framework.

\begin{description}
\item[Ranking candidates.] Assume the goal is to find a ranking of the candidates~\cite{dwo-kum-nao-siv:c:rank-aggregation,proprank} instead of simply picking a subset. The ranking should reflect the preferences of the voters expressed over the individual candidates. This setting can be represented in our general model as follows. For each pair of candidates, $c_1$ and $c_2$, we introduce an auxiliary candidate $c_{1, 2}$. Intuitively, selecting  $c_{1, 2}$ would correspond to putting $c_1$ before $c_2$ in the returned ranking. The feasibility constraints would ensure that the selected auxiliary candidates correspond to a transitive, asymmetric, and complete relation on original candidates.\footnote{It remains to specify the voters' preferences over the auxiliary candidates. The most natural way is to construct an approval-based preference profile, and to assume that a voter approves $c_{1, 2}$ if she prefers candidate $c_1$ over $c_2$ in the original preference profile. This approach would be compatible with preference profiles consisting of weak partial orders.} The model also applies to collective scheduling~\cite{sko-rza-pas:collective_scheduling}---the candidates would correspond to jobs to be scheduled, and the  constraints allow to incorporate additional dependencies between the jobs.
\item[Committee elections with negative votes.] Consider a model where the voters are allowed to express negative feelings towards candidates~\cite{bau-den:mav_trichotomous,zhou2019parameterized}. This scenario can be modelled by introducing auxiliary candidates and adding appropriate constraints. For each candidate $c$ we can add a virtual candidate $\bar{c}$ which corresponds to not-selecting $c$. A voter approves $\bar{c}$ if she voted against $c$ in the original election. The feasibility constraints would ensure that we never select $c$ and $\bar{c}$ together.

\item[Judgment Aggregation.] Here the goal is to find a valuation of propositional variables that satisfies a given set of propositional formulas~\cite{lis-pol:judgment-aggregation,End15JA}. The valuation should take into account the opinions of the voters with respect to which of the variables should be set true, and which of them should be set false.  We can represent this setting by adding, for each propositional variable $x$, two candidates, $c_{x, \mathrm{T}}$ and $c_{x, \mathrm{F}}$, corresponding to setting the variable to true and to false, respectively. The propositional formulas can be incorporated as feasibility constraints.
\end{description}

In \Cref{sec:matroid} we give examples showing that the aforementioned constraints are not matroids.

Our model is also closely related to voting in combinatorial domains~\cite{lan-via:combinatorial-domains}.

\section{Definition of Proportionality}

In this section, we formulate our main definition that capture the idea of group fairness. We start by defining the base axiom of proportionality, called Base Extended Justified Representation (BEJR). This axiom already implies strongest notions of proportionality in the more specific models. Specifically, it implies Extended Justified Representation in the model of committee elections~\citep{justifiedRepresentation}, proportionality for cohesive groups in the model of public decisions~\cite{sko-gor:proportional_public_decisions} and strong EJR~\citet{cha-goe-pet:seq-decision-making} in the context of sequential decision making~\cite{lac:perpetual-voting}. It is always satisfiable, and has an intuitive interpretation.

Next, we compare our axiom to the recent definition of Restrained EJR by \citet{mav-mun-she:committees-with-constraints}. We show that in many natural settings---for example, in the case of demographic constraints---our axiom provides considerably stronger guarantees. There are, however cases, where Restrained EJR is not implied by our Base EJR. This discussion provides additional insights and leads to our main definition of EJR. The main definition of EJR is very similar to its base counterpart, thus all our intuitive explanations of the base EJR cary over. We believe it is most instructive to understand the base axiom first rather than to go directly to our main definition of EJR.

The main definition always exists for matroid constraints and can be satisfied by natural extensions of the known election rules. For non-matroid constraints the definition might be too strict---for example, it cannot be always satisfied together with Pareto optimality. This is in contrast to the case of base EJR, which never contradicts Pareto optimality.

\subsection{The Base Notion of Proporitonality}

Let us start by introducing the base variant of our main axiom. Next, we will provide its intuitive explanation and will show a few structural properties of the proposed definition.

\begin{definition}[Base Extended Justified Representation (BEJR)]\label{def:ejr}
  Given an election $E = (C, N, \feasibles, \mathcal{A})$ we say that a group of voters $S \subseteq N$ \myemph{deserves} $\ell$ candidates if for each feasible set $T \in \feasibles$ either there exists $X \subseteq \bigcap_{i\in S} A_i$ with $|X| \geq \ell$ such that $T \cup X \in \feasibles$, or the following inequality holds:
\begin{align}\label{eq:ejr_condition}
\frac{|S|}{n} > \frac{\ell}{|T| + \ell} \text{.}
\end{align}
We say that a feasible outcome $W \in \feasibles$ of an election $E = (C, N, \feasibles, \mathcal{A})$ satisfies \myemph{base extended justified representation (BEJR)} if for each $\ell \in \naturals$ and each group of voters $S\subseteq{}N$ that deserves $\ell$  candidates there exists a voter $i \in S$ that approves at least $\ell$ candidates in $W$, i.e., $u_i(W) \geq \ell$. \hfill $\lrcorner$
\end{definition}

Note that for $|S| \neq n$ and $T\neq \emptyset$ condition \eqref{eq:ejr_condition} can be equivalently written as:
\begin{align}\label{eq:ejr_condition2}
\frac{|S|}{n - |S|} > \frac{\ell}{|T|} \text{.}
\end{align}
The latter formulation might look a bit more intuitive, but we will mainly use the former one, since it does not require considering the case of division by zero separately.

Let us intuitively explain \Cref{def:ejr}. Consider a group of voters $S \subseteq N$ and let us have a closer look at the condition saying that this group deserves~$\ell$ candidates. Why giving $\ell$ candidates to $S$ can be possibly wrong?
The main reason is it may prohibit us from selecting candidates that are supported by other voters, namely the voters from $N \setminus S$. Consider a set $T$ that is supported by those from $N \setminus S$. If $T \cup X \in \feasibles$ then giving $\ell$ candidates to $S$ does not prohibit selecting $T$; we can safely give $\ell$ candidates to $S$ while satisfying the claim of the remaining voters.  If $T \cup X \notin \feasibles$ then we are in \eqref{eq:ejr_condition}; for the sake of this explanation consider the (almost equivalent) formulation given in \eqref{eq:ejr_condition2}. This formulation reads as follow: proportionally to its size the claim of the group $S$ to $\ell$ candidates is stronger than the claim of the remaining voters to set $T$. Thus, such set $T$ cannot be used as an evidence discouraging us from giving $\ell$ candidates to $S$.

Yet another equivalent formulation of the condition in \Cref{def:ejr} is the following. A nonempty group of voters $S \subseteq N$ deserves $\ell$ candidates if for each feasible set $T \in \feasibles$ with
\begin{align}\label{eq:ejr_condition3}
|T| < \ell \cdot \frac{n-|S|}{|S|}
\end{align}
there exists $X \subseteq \bigcap_{i\in S} A_i$ with $|X| \geq \ell$ such that $T \cup X \in \feasibles$. This condition intuitively reads as follows: $S$ deserves $\ell$ candidates if they can complete each reasonable suggestion of the other voters, $T$, with $\ell$ commonly approved candidates.


\begin{remark}
We note that in \Cref{def:ejr} we might w.l.o.g.\ assume that $X \in \feasibles$. Indeed, if $T \cup X \in \feasibles$, then in particular $X \in \feasibles$ since $\feasibles$ is closed under inclusion.  \hfill $\lrcorner$
\end{remark}
\begin{remark}
If a group of voters $S$ deserves $\ell$ candidates then in particular, there must exist a feasible set $X \subseteq \bigcap_{i\in S} A_i$ with $|X| \geq \ell$. This follows from the observation that for an empty set $T = \emptyset$ condition \eqref{eq:ejr_condition} is never satisfied.  \hfill $\lrcorner$
\end{remark}


Let us now illustrate our definition through a couple of examples. This will also provide intuition on why our definition generalises the analogous definitions in the more specific models.

\begin{example}\label{ex:commitee_elections}
Consider the model of committee elections with approval utilities. Assume the goal is to select a subset of $k = 10$ candidates, and consider a group $S$ consisting of $30\%$ of voters who jointly approve three candidates, $c_1, c_2$, and $c_3$.
Indeed, consider a set $T \subseteq C$, and observe that $X = \{c_1, c_2, c_3\}$ always satisfies the conditions from \Cref{def:ejr}. If $|T| \leq 7$ then $T \cup X \in \feasibles$. Otherwise,
\begin{align*}
\frac{3}{|T| + 3} \leq \frac{3}{8 + 3 } < \frac{3}{10} = \frac{|S|}{n} \text{.}
\end{align*}
Thus, BEJR implies that there exists a voter from $S$ that approves at least three out of ten selected candidates.\hfill $\lrcorner$
\end{example}

\begin{example}\label{ex:public-decisions}
Consider the model of public decisions, and a group $S$ of $30\%$ of voters who have the same opinion with respect to some $p$ issues. We will prove that this group deserves $\lfloor 0.3 \cdot p \rfloor$ candidates (here, decisions). Consider a set $T \in \feasibles$. If $|T| \leq p - \lfloor 0.3 \cdot p \rfloor$, then we can find $\lfloor 0.3 \cdot p \rfloor$ decisions that $S$ agrees on, and we can add them to $T$ so that the set is feasible. Otherwise, we get:
\begin{align*}
\frac{\lfloor 0.3 \cdot p \rfloor}{|T| + \lfloor 0.3 \cdot p \rfloor} < \frac{\lfloor 0.3 \cdot p \rfloor}{p - \lfloor 0.3 \cdot p \rfloor + \lfloor 0.3 \cdot p \rfloor} = \frac{\lfloor 0.3 \cdot p \rfloor}{p} \leq \frac{ 0.3 \cdot p}{p} = 0.3 \leq \frac{|S|}{n} \text{.}
\end{align*}
Thus, in both cases the conditions in \Cref{def:ejr} are satisfied.

This example also shows why we cannot treat the constraints separately, and why it is not enough to consider proportionality independently within each constraint. Indeed, if we did so, then in the model of public decisions would need to guarantee the proportionality with respect to each single decision only. In effect, a group of less than 50\% of voters would not be guaranteed to have any influence on the outcome, even if they agreed with respect to all the issues. \hfill $\lrcorner$
\end{example}

The same reasoning can be used to formally prove that  \Cref{def:ejr} generalises the classic definition of EJR from the literature on committee elections~\citep{justifiedRepresentation}, and that it corresponds to the definition of proportionality for cohesive groups in the model of public decisions (Definition 7 in~\cite{sko-gor:proportional_public_decisions}). Finally, our definition implies the axiom of strong EJR by~\citet{cha-goe-pet:seq-decision-making} in the context of sequential decision making~\cite{lac:perpetual-voting} (a weaker variant of strong EJR is strong PJR; this axiom has been also considered by \citet{bulteau2021jr}, but they used the name ``some periods intersection PJR'').
Let us consider yet another example.

\begin{example}
Consider the model of committee elections with disjoint attributes.
Assume that $z = 2$ and so $C = C_1 \cup C_2$. Assume that our goal is to select exactly 10 candidates from $C_1$ and exactly 20 candidates from $C_2$. Thus, $\lquota{1} = \uquota{1} = 10$, $\lquota{2} = \uquota{2} = 20$, and $k = 30$. Assume further that there are enough candidates in each set, e.g., $|C_1| = |C_2| = 100$.

Let $S$ consists of 41\% of all the voters, who jointly approve some 11 candidates from $C_2$. These voters deserve 8 candidates. Indeed, let $X$ be a set of 8 candidates jointly approved by $S$. If $|T| < 13$ then $T \cup X \in \feasibles$; otherwise:
\begin{align*}
\frac{8}{|T| + 8}  \leq \frac{8}{13 + 8} = \frac{8}{21} < \frac{|S|}{n} \text{.}
\end{align*}
Now assume that $S$ additionally approves 4 candidates from $C_1$. Then $S$ deserves 10 candidates. Indeed, consider two cases. If $|T \cap C_1| \leq 6$ then we can add to $X$ four candidates from $C_1$ without violating the constraints. In order to prevent adding 6 candidates from $C_2$, it must hold that $|T \cup C_2| \geq 15$. But then:
\begin{align*}
\frac{10}{|T| + 10}  \leq \frac{10}{25} < \frac{|S|}{n} \text{.}
\end{align*}
On the other hand, if $|T \cap C_1| > 6$ then we observe that it also must hold that $|T \cap C_2| > 11$ (as otherwise we could add to $X$ ten candidates from $C_2$), and so $|T| \geq 18$. Thus, also in this case the condition  \eqref{eq:ejr_condition} from \Cref{def:ejr} holds. \hfill $\lrcorner$
\end{example}

The main feature which makes our definitions powerful is that they are always satisfiable, independently of the specific types of constraints or voters' preferences.

\begin{theorem}\label{thm:ejr:existence}
For each election there exists an outcome satisfying Base EJR.
\end{theorem}

The proof of \Cref{thm:ejr:existence} follows from a more general result, namely from \Cref{thm:fjr:existence} in \Cref{sec:extensions}.

Finally, note that BEJR also implies that the average utility of the voters from the group $S$ is considerably high. This is already known in the context of committee elections~\cite{pjr17,skowron:prop-degree}, and below we generalise this result to matroid constraints.

\begin{proposition}\label{prop:ejr-average-sat}
Consider an election with matroid feasibility constraints, and let $W$ be an outcome satisfying BEJR. Then, for each group of voters $S$ deserving $\ell$ candidates, the average number of candidates from $W$ that the voters from $S$ approve is at least:
\begin{align}
  \frac{1}{|S|}\sum_{i \in S} |A_i \cap W| \geq \frac{\ell-1}{2} \text{.} \label{eq:prop3}
\end{align}
This estimation is tight up to the constant of one.
\end{proposition}
\begin{proof}
Consider a group of voters $S \subseteq N$ that deserves $\ell$ candidates.

Consider a group $S' \subseteq S$ with $|S'| \geq |S| - i \cdot \frac{|S|}{\ell}$ for some natural number $i \in [\ell]$. We will first show that $S'$ deserves $\ell - i$ candidates. Fix a feasible subset of candidates $T \in \feasibles$.  Let us remove~$i$ arbitrary candidates from $T$, and denote the so-obtained subset as $T'$; if $|T| < i$, then we simply set $T' = \emptyset$.  Let us consider two cases. First, assume that there exists $X \subseteq \bigcap_{i \in S} A_i$ of size $\ell$ such that $X \cup T' \in \feasibles$. Let $p = |T' \cap X|$. Clearly:
\begin{align*}
|X \cup T'| - |T| \geq \ell + |T'| - p - |T| \geq  \ell - p + |T| - i - |T| =  \ell - p - i \text{.}
\end{align*}
Then, by the exchange property \ref{cond:exchange-property} applied to $X \cup T'$ and $T$, we get that we can add $\ell - p - i$ candidates from $X$ to $T$ and the so obtained set would be feasible. Consequently, there exist a set $X'\subseteq X$ of size $\ell - i$ such that $X' \cup T \in \feasibles$ and so the condition from \Cref{def:ejr} is satisfied for $S'$.

Second, assume that for each $X \subseteq \bigcap_{i \in S} A_i$ of size $\ell$ we have $X \cup T' \notin \feasibles$. Then, since $S$ deserves $\ell$ candidates, we get that:
\begin{align*}
\frac{|S|}{n} > \frac{\ell}{|T'| + \ell} \text{.}
\end{align*}
This, in particular means that $T' \neq \emptyset$ and so $|T'| = |T| - i$. Consequently:
\begin{align*}
\frac{|S'|}{n} \geq \frac{|S| }{n} \cdot \left(1 - \frac{i}{\ell}\right) > \frac{\ell}{|T'| + \ell}\cdot \frac{\ell - i}{\ell} =  \frac{\ell - i}{|T'| + \ell} = \frac{\ell - i}{|T| + \ell - i} \text{.}
\end{align*}
This again shows that the condition from \Cref{def:ejr} is satisfied for $S'$.

Thus, by EJR we know that there must exist a voter $v_1$ who approves at least $\ell$ candidates in the outcome $W$. We can apply EJR to $S \setminus \{v_1\}$ and infer that there exists a voter $v_2$ with a given number of approved candidates in $W$, and so on. Altogether, we get that at least one voter approves $\ell$ candidates at least $\left\lfloor\frac{|S|}{\ell} \right\rfloor$ voters approve $\ell - 1$ candidates, and so on. Thus:
\begin{align*}
\frac{1}{|S|}\sum_{i \in S} |A_i \cap W| &\geq \frac{1}{|S|}\cdot \left(\ell + \sum_{j = 1}^{\ell-1} \left\lfloor\frac{|S|}{\ell} \right\rfloor  \cdot j \right) \geq  \frac{1}{|S|}\cdot \left(\sum_{j = 1}^{\ell-1} \frac{|S|}{\ell}  \cdot j \right) \\
&= \frac{1}{\ell}\cdot \sum_{j = 1}^{\ell-1}j = \frac{1}{\ell} \cdot \frac{(\ell-1)\ell}{2} =  \frac{\ell-1}{2} \text{.}
\end{align*}

This estimation is tight even for committee elections. For example, it is known that the method of equal shares satisfies the axiom of extended justified representation~\cite{pet-sko:laminar} and that the average number of candidates approved by the voters from $S$ might be equal to $\frac{\ell+1}{2}$~\cite{lac-sko:multiwinner-book}.
\end{proof}

\subsection{Base EJR versus Restrained EJR}\label{sec:base_ejr_vs_restrained_ejr}

Recently, \citet{mav-mun-she:committees-with-constraints} have considered the model of committee elections with constraints. While they mainly focused on the notion of the core, they also proposed yet another variant of EJR that applies to the model with constraints. Our work has been done independently\cite{mas-pie-sko:group-fairness-arxiv}, yet it is important to compare the two definitions. First let us recall the definition by \citet{mav-mun-she:committees-with-constraints}. We slightly simplified the original definition---this was possible, because we assume the feasibility constraints are closed under inclusion.\footnote{The assumption that the feasibility constraints are closed under inclusion is very mild. In fact, it would be only restraining if we assumed that some candidates can generate negative utilities. This, however, is not the case for dichotomous utilities nor for general monotone utility functions that we consider in \Cref{sec:geenral-utility-functions}.}

\begin{definition}[Restrained EJR for Approval Utilities~{\cite[Definition 3.1]{mav-mun-she:committees-with-constraints}}]\label{def:restrainedEJR}
  Let $k$ be the maximum size of feasible outcomes, $ k = \max_{W \in \feasibles}|W|$.
  Let $k' = \left\lfloor \frac{|S|}{n} k\right\rfloor$ be the \emph{endowment} of $S\subseteq N$.
  We say that $S\subseteq N$ with endowment $k'$ is a \emph{blocking coalition} for some $W\in \feasibles$ if it satisfies the following:
  For all feasible outcomes $\hat{W} \subseteq W$ with $|\hat{W}| \leq k - k'$, there exists $W'$ with $|W'| \leq k'$ such that:
\begin{enumerate}
\item $T = \hat{W} \cup W' \in \feasibles$, and
\item $\left|\bigcap_{i\in S} A_i \cap T\right| \geq \max_{i \in S} u_i(W) + 1$.
\end{enumerate}
  An outcome $W \in \feasibles$ satisfies \myemph{restrained EJR} if there is no blocking coalition of voters $S \subseteq N$.\hfill $\lrcorner$
\end{definition}

We first complement the work of \citet{mav-mun-she:committees-with-constraints} by showing two important results concerning the satisfiability of the axiom. The proofs of the two theorems are given in \Cref{appx:restrained-ejr}.

\begin{theorem}\label{thm:restrained-ejr-exists}
For each election there exists an outcome satisfying Restrained EJR.
\end{theorem}

\begin{theorem}\label{thm:restrained-ejr-and-pareto}
There exists no selection rule that satisfies Restrained EJR and Pareto Optimality.
\end{theorem}

For matroid constraints Restrained EJR can be satisfied together with Pareto Optimality---this follows from \Cref{thm:pav-and_ejr} in \Cref{sec:pav}. This suggests that Restrained EJR might be better suited to the special case of matroid constraints. The axiom of Base EJR, on the other hand, does not exclude Pareto optimality even in the most general variant of the model. 

In fact there are even more substantial differences between Base and Restrained EJR, which we illustrate through the following examples.

\begin{example}[Restrained EJR fails intuitive proportionality for demographic constraints]\label{ex:restrained-ejr-weak}
Consider the model of committee elections with disjoint attributes. Assume our goal is to select 100 candidates, 50 of which are men and 50 are women. Thus, $z = 2$, $C = C_1 \cup C_2$ (for example, $C_1$ consists of men and $C_2$ of women), $k = 100$, $\lquota{1} = \uquota{1} = \lquota{2} = \uquota{2} = 50$. Assume there is a group $S$ consisting of 50\% of voters who vote for some $50$ candidates from $C_1$, say for the candidates $m_1, \ldots, m_{50}$. The remaining 50\% of voters vote for some other 50 candidates from $C_1$ ($m_{51}, ..., m_{100}$). Then the committee consisting of $m_1, ..., m_{50}$ and some arbitrary 50 candidates from $C_2$ satisfies Restrained EJR. Indeed, the group $S$ is not a blocking coalition because for $\hat{W} = \{m_1, ..., m_{50}\}$ no candidates approved by $S$ can be added to $\hat{W}$. At the same time, Base EJR implies that at least half of the candidates approved by $S$ must be selected. \hfill $\lrcorner$
\end{example}

We believe that \Cref{ex:restrained-ejr-weak} illustrates a serious limitation of Restrained EJR, since elections with such simple demographic constraints are perhaps the most natural applications of the general model with constraints. At the same time, there are cases where Restrained EJR provides stronger guarantees than Base EJR, as illustrated bellow.

\begin{example}\label{ex:base-ejr-weak}
Once again, consider a model of committee elections with two disjoint attributes. Thus, $z = 2$, and $C = C_1 \cup C_2$. We need to select 20 candidates in total, at most 14 from each set, thus $k = 20$, $\lquota{1} = \lquota{2} = 0$, and $\uquota{1} = \uquota{2} = 14$. The group $S$ of 50\% of voters approves 10 candidates from $C_1$---let us call them $a_1, ..., a_{10}$. The remaining voters approve some other 3 candidates from $C_1$ (call them $a_{11}$, $a_{12}$, $a_{13}$) and some 10 candidates from $C_2$ (say, $b_1$, \ldots, $b_{10}$).
Consider the committee $W = \{a_1, ..., a_7, a_{11}, ..., a_{13}, b_{1}, ..., b_{10}\}$. This committee satisfies Based EJR but fails Restrained EJR; indeed $S$ is a blocking coalition. A committee that satisfies Restrained EJR must contain all 10 candidates approved by $S$.
\end{example}

Let us have a closer look at \Cref{ex:base-ejr-weak}. We can see that the voters from $S$ and $N\setminus S$ disagree on which candidates from $C_1$ should be selected; there is no disagreement with respect to the candidates from $C_2$, since the voters from $S$ do not like any candidate there. This is why, according to Base EJR, the group $S$ is guaranteed half of the candidates from the set they care about, that is 7 candidates. Restrained EJR on the other hand guarantees that $S$ gets half of the whole outcome, that is 10 candidates. While both interpretations are reasonable, we believe that  in this case the stronger guarantee provided by Restrained EJR is more compelling. This motivates us to slightly strengthen the definition of Base EJR.

\subsection{The Main Definition of Proportionality}

We are now ready to introduce our main definition that strengthens Base and Restrained EJR.

\begin{definition}[Extended Justified Representation (EJR)]\label{def:main-ejr}
  Given an election $E = (C, N, \feasibles, \mathcal{A})$, and an outcome $W \in \feasibles$ we say that a group of voters $S \subseteq N$ deserves $\ell$ candidates in $W$ if for each set $T\subseteq W$ either there exists $X \subseteq \bigcap_{i\in S} A_i$ with $|X| \geq \ell$ such that $T \cup X \in \feasibles$, or the following inequality holds:
\begin{align}\label{eq:tejr_condition}
\frac{|S|}{n} > \frac{\ell}{|T| + \ell} \text{.}
\end{align}
We say that a feasible outcome $W \in \feasibles$ satisfies \myemph{extended justified representation (EJR)} if for each group of voters $S \subseteq N$ deserving $\ell$ candidates in $W$ there is a voter $i \in S$ such that $u_i(W) \geq \ell$. \hfill $\lrcorner$
\end{definition}

Intuitively, in the definition of EJR compared to Base EJR we consider only sets $T$ that are subsets of $W$. The justification of this change is the following: the sets $T$ that can prohibit us from giving $\ell$ candidates to $S$ (according to Base EJR) can consist of candidates that have very little support among the voters, and thus would not be selected anyway. This motivates focusing on sets $T$ that are being considered for selection, namely the sets contained in the outcome at hand, $W$. 

Another difference is that in Base EJR, the entitlement of a group of voters $S\subseteq N$ depends solely on the preferences of the voters from $S$. On the contrary, according to EJR the group might be guaranteed different number of candidates depending on what is the winning outcome, hence depending on the preferences of the remaining voters. We see it as a small advantage of Base EJR, as it allows each group of voters understand what influence on the outcome they will have irrespectively of the preferences of others.

\Cref{def:main-ejr} can be equivalently written in a way that more closely resembles Restrained EJR.

\begin{definition}[Extended Justified Representation (EJR); equivalent definition]\label{def:ejr:equiv}
  Given an election $E = (C, N, \feasibles, \mathcal{A})$, we say that a group of voters $S \subseteq N$ \myemph{deviates} in an outcome $W\in\feasibles$ if for $\ell= \max_{i \in S} u_i(W)+ 1$ and for each set $T\subseteq W$ either there exists $X \subseteq \bigcap_{i\in S} A_i$ with $|X| \geq \ell$ such that $T \cup X \in \feasibles$, or the following inequality holds:
\begin{align}\label{eq:tpjr_condition}
\frac{|S|}{n} > \frac{\ell}{|T| + \ell} \text{.}
\end{align}
We say that a feasible outcome $W \in \feasibles$ of an election $E = (C, N, \feasibles,\mathcal{A})$ satisfies extended justified representation (EJR) if no group of voters $S\subseteq N$ deviates in $W$.
\hfill $\lrcorner$
\end{definition}

Below we prove that EJR is indeed stronger than both Base EJR and Restrained EJR.

\begin{proposition}\label{prop:strongerEJR}
 Let $E = (C, N, \feasibles, \mathcal{A})$ be an election and let $W\in \feasibles$.
 If $W$ satisfies EJR then $W$ satisfies both Base EJR and Restrained EJR.
\end{proposition}
\begin{proof}
Consider an outcome $W$ that satisfies EJR. First, we show that $W$ satisfies Base EJR. Towards a contradiction suppose there is a group of voters $S\subseteq N$ that deserves $\ell$ candidates, but for all voters $i\in S$ we have $u_i(W)<\ell$. Let $\ell' = \max_{i \in S} u_i(W) + 1$.
Since $S$ deserves $\ell$ candidates we know that for each $T\in \feasibles$ either there exists $X \subseteq \bigcap_{i \in S} A_i$ with $|X|\ge \ell \geq \ell'$ such that $T\cup X \in \feasibles$, or
\begin{align*}
\frac{|S|}{n} > \frac{\ell}{|T| + \ell}\geq \frac{\ell'}{|T| + \ell'} \text{.}
\end{align*}
This in particular holds for every $T \subseteq W$. Thus, $S$ deviates in $W$, a contradiction.

 Second, we show that $W$ satisfies Restrained EJR. For the sake of contradiction, assume the opposite.
 Let $k$ be the maximum size of committee in $\feasibles$, and consider a blocking coalition $S$.
  Let $k' = \left\lfloor \frac{|S|}{n} k\right\rfloor$ and let $\ell = \max_{i \in S} u_i(W) + 1$. From the condition on the blocking coalition applied to $\hat{W} = \emptyset$ we infer that $\ell \leq k'$. We will prove that then $S$ deviates in $W$.
Consider a set $T \subseteq W$. If $|T| > k - k'$, then we have:
\begin{align*}
\frac{\ell}{\ell + |T|} \leq \frac{k'}{k' + |T|} < \frac{k'}{k' +  k - k'} = \frac{k'}{k} \leq \frac{|S|}{n} \text{.}
\end{align*}
Thus, Condition \eqref{eq:tejr_condition} is satisfied.
If $|T| \leq k - k'$ then  as $S$ is a blocking coalition, by setting $\hat{W} = T$ we infer that there exists $W'$ such that $\hat{W} \cup W' \in \feasibles$ and $|\bigcap_{i\in S} A_i \cap (\hat{W} \cup W')| \geq \ell$. Thus, there exists $X \subseteq \bigcap_{i\in S} A_i$ with $|X| \geq \ell$ such that $T \cup X \in \feasibles$. Consequently, $S$ deviates in $W$.
\end{proof}

It is an open problem whether there always exists an outcome satisfying EJR. However, we know that this axiom is always satisfiable for matroid constraints. For non-matroid constraints the axiom might contradict Pareto optimality (see \Cref{thm:restrained-ejr-and-pareto}), which suggests that considering the weaker variant of Base EJR is more appropriate in this case.

Finally, we note that an analogous result to \Cref{prop:ejr-average-sat} also holds for EJR, that is EJR implies high average utility for the groups of voters with cohesive preferences. 

\section{Proportional Approval Voting}\label{sec:pav}

In this section we consider a natural extension of Proportional Approval Voting (PAV) to our model with general constraints. We characterise the structure of feasibility constraints for which PAV satisfies our notions of proportionality. Thus, our characterisation precisely identifies the elections for which it is appropriate to use Proportional Approval Voting.

\begin{definition}[Proportional Approval Voting (PAV)]
  Given an election $E = (C, N, \feasibles,\mathcal{A})$, we define the PAV score of an outcome $W \subseteq C$ as:
\begin{align*}
\score_{\pav}(W) = \sum_{i \in N} \H(|W \cap A_i|), \qquad \text{where~} H(k) = \sum_{j = 1}^k \nicefrac{1}{j} \text{.}
\end{align*}
Proportional Approval Voting (PAV) selects a feasible outcome with the maximal PAV score. \hfill $\lrcorner$
\end{definition}

PAV has excellent properties pertaining to proportionality in the model of committee elections~\cite{justifiedRepresentation, AEHLSS18, pet-sko:laminar, BGPSW19, lac-sko:multiwinner-book} and public decision~\cite{sko-gor:proportional_public_decisions}. We will now prove that PAV exhibits good properties also when applied to the model with more general constraints---precisely, that PAV satisfies EJR if and only if the feasibility constraints have a matroid structure.
One potential drawback of PAV is that it is NP-hard to compute, even for committee elections~\cite{owaWinner,azi-gas-gud-mac-mat-wal:c:multiwinner-approval}. However, a polynomial-time local-search algorithm for PAV preserves its strong proportionality guarantees~\cite{AEHLSS18,kra-elk:local-search-PAV}. It will follow from our proofs that this is also the case for the more general case of matroid constraints.

\begin{theorem}\label{thm:pav-and_ejr}
PAV satisfies EJR for all elections with matroid constraints. For each non-matroid feasibility constraints there is an election where PAV fails Base EJR.
\end{theorem}
\begin{proof}
First, we prove that if the election has matroid constraints, then PAV satisfies EJR.

    Suppose that the statement does not hold for some election $E$. Let $W$ be an outcome returned by PAV for $E$. Let $S\subseteq N$ be a group of voters that deviates in $W$, and let $\ell=\max_{i \in S} u_i(W) + 1$.
    Let us denote by $A_S$ the set $\bigcap_{i\in S} A_i$. Further, let \[W'=\left\{c\in W : \exists a_c\in A_S \text{ s.t.\ } (W\setminus\{c\})\cup\{a_c\} \in \feasibles\right\}.\]
    Intuitively, for each candidate $c\in W'$ we can swap $c$ with $a_c$, and after such a swap the outcome will still be feasible. Of course, for each candidate $c\in W\setminus W'$ we can also swap $c$ with herself (such a swap does not change the outcome).

    Let us denote by $\Delta(c, c')$ the change in the PAV score obtained by swapping some $c\in W$ with some $c'\in C$. We know that for each such pair of candidates we have $\Delta(c, c') \leq 0$ (as otherwise $W$ would not be an outcome maximizing PAV score). Let us estimate the following expression, which describe the sum of the score change if we swap each time a different element of $W$:
    \begin{equation*}
        \sum_{c\in W'} \Delta(c, a_c) + \sum_{c\notin W'} \Delta(c, c) \text{.}
    \end{equation*}
    Swapping a pair of candidates can be viewed as a two step process, where first we remove one candidate from $W$ and then we add one.
    Let us first estimate the sum of decreases of the PAV score due to removing candidates. As we removed each candidate exactly once, for each voter having $x > 0$ representatives in $W$, we can subtract $x$ times the score of $\nicefrac{1}{x}$. Let us denote by $S_0\subseteq S$ the subset of voters from $S$ who have no representatives in $W$. Hence, the total decrease may be equal to $n-|S_0|$ at most.

   \allowdisplaybreaks
    Now after additions of new candidates to the committee, the PAV score increases in total by at least:
    \begin{align*}
       & \sum_{i\in S\setminus S_0} \left(\frac{|W'\setminus A_i|}{|W\cap A_i|+1}  + \frac{|W\cap A_i|}{|W\cap A_i|}\right) + \sum_{i\in S_0} \frac{|W'\setminus A_i|}{|W\cap A_i|+1}\\
       &\geq \sum_{i\in S} \left(\frac{|W'\setminus A_i|}{|W\cap A_i|+1} + 1\right) - |S_0|\\
        &= \sum_{i\in S} \frac{|W'\setminus A_i| + |W\cap A_i|+1}{|W\cap A_i|+1} - |S_0|\\
        &\geq \sum_{i\in S} \frac{|W'| + |(W\setminus W')\cap A_i|+1}{\ell} - |S_0|\\
        &\geq \sum_{i\in S} \frac{|W'| + |(W\setminus W')\cap A_S|+1}{\ell} - |S_0|\\
        &\geq |S|\cdot \frac{|W'\setminus A_S| + |W\cap A_S|+1}{\ell} - |S_0|.
        \end{align*}
    Intuitively, for each voter $i\in S$, if the removed candidate was from $W\cap A_i$, we add the $(|W\cap A_i|)$th candidate supported by her, otherwise we add the $(|W\cap A_i|+1)$th candidate supported by her.

    Finally, we have that:
    \begin{equation*}
        |S|\cdot \frac{|W'\setminus A_S| + |W\cap A_S|+1}{\ell} - |S_0| - (n - |S_0|) \leq 0.
    \end{equation*}
    \begin{equation}\label{eq:pav-upperbound}
        \frac{|S|}{n} \leq \frac{\ell}{|W'\setminus A_S| + |W\cap A_S|+1}.
    \end{equation}

    Consider now set $T$ obtained in the following way:
    \begin{enumerate}[(1)]
        \item first, we set $T = W'$,
        \item second, we remove from $T$ all the candidates from $W'\cap A_S$ and some arbitrary additional $\ell-1-|W\cap A_S|$ candidates.
    \end{enumerate}

    We will prove that $S$ cannot propose a subset $X\subseteq A_S$ of size $\ell$ such that $T\cup X\in \feasibles$. Indeed, otherwise from the \ref{cond:exchange-property} applied to sets $W'$ and $T\cup X$ we would have that there exists a candidate $c\in X\setminus W'$ such that $W'\cup\{c\}\in \feasibles$. We distinguish two cases. Either $W=W'$, but then $W\cup \{c\}$ is a set of greater score than $W$, a contradiction. Or $W\setminus W'\neq \emptyset$, but then from \Cref{lem:one-swap} applied to $W$, $W\setminus W'$ and $c$, we get that there exists a candidate $c'\in W\setminus W'$ such that $W\setminus \{c'\}\cup \{c\}\in \feasibles$. But from the definition of set $W'$ there need to hold that $c'\in W'$, a contradiction.

    Since we have proved that group $S$ cannot propose any committee $X\subseteq A_S$ of size $\ell$ such that $T\cup X\in \feasibles$, it needs to hold:
    \begin{equation*}
        \frac{|S|}{n} > \frac{\ell}{|T| + \ell}.
    \end{equation*}

    From the construction of set $T$, we know that:
    \begin{align*}
        |T|  &= |W'| - |W'\cap A_S| - (\ell - 1 - |W \cap A_S|)\\
              &= |W'\setminus A_S| + |W\cap A_S| - \ell + 1 \text{.}
    \end{align*}
    Therefore:
    \begin{equation}\label{eq:pav-lowerbound}
        \frac{|S|}{n} > \frac{\ell}{|W'\setminus A_S| + |W\cap A_S|+1}.
    \end{equation}
    Joining  \eqref{eq:pav-lowerbound} and \eqref{eq:pav-upperbound} we obtain a contradiction, which completes the first part of the proof.

    Now we prove that if the feasibility constraints are not a matroid, then PAV violates Base EJR. Suppose that for given constraints, there exist sets $X, Y\in \feasibles$ such that $|X| < |Y|$ and $X\cup\{c\}\notin \feasibles$ for all $c\in Y\setminus X$. Denoted such a pair $X,Y$ as \emph{witness}. First, observe that for a witness the following is true: $X\neq \emptyset$, $X\nsubseteq Y$, and $|Y\setminus X| \geq 2$. All the statements follow from the fact that $\feasibles$ is closed under inclusion. Now, among all such witnesses, consider a one that first minimize $|X|$ and second minimize $|Y\setminus X|$. Denote by $\ell = |X|>1$.

    Consider now the following construction. Let $n$ be a multiple of $3\cdot (\ell+1)$. We have a group $S_1$ of $\nicefrac{\ell}{\ell+1}\cdot n + 1$ voters approving exactly $X$ and a group $S_2$ of the remaining $\nicefrac{n}{\ell+1} - 1$ voters approving exactly $Y\setminus X$.

    Let us first show that $S_1$ deserves $\ell$ candidates. Indeed, for $T=\emptyset$ this group can propose the feasible set $X$ of size $\ell$ they jointly approve. On the other hand, if $|T| \geq 1$, then we have:
    \begin{align*}
    \frac{|S_1|}{n} = \frac{\ell}{\ell+1} + \frac{1}{n} > \frac{\ell}{\ell+1} \geq \frac{\ell}{\ell+|T|} \text{.}
    \end{align*}
    Since $S_1$ deserves $\ell$ candidates, Base EJR may be satisfied only if all the candidates from $X$ are elected. Then, directly from our assumptions, no candidate from $Y \setminus X$ can be elected.

    Now consider set $X$ with one candidate $c'\in X$ removed. Then there exists $c_1\in Y\setminus X$ such that $(X\setminus\{c'\})\cup\{c_1\}\in \feasibles$, since otherwise $(X\setminus\{c'\}, Y)$ would be a smaller witness. Similarly, we conclude that there exists also $c_2 \in Y\setminus X$, $c_2 \neq c_1$ such that $(X\setminus \{c'\})\cup\{c_1, c_2\}\in \feasibles$, since otherwise $((X\setminus\{c'\})\cup\{c_1\}, Y\setminus\{c_1\})$ would be a smaller witness (indeed, we have $|(X\setminus\{c'\})\cup\{c_1\}|< |Y\setminus\{c_1\}|$ as $|X|<|Y|+1$).

    Note that all the candidates outside of $X \cup Y$ contribute $0$ PAV points to the final score, hence if Base EJR is satisfied, the PAV score of the winning outcome cannot be higher than the score of outcome $X$. However, the score of outcome $(X\setminus \{c'\})\cup\{c_1, c_2\}$ is higher---compared to $X$, group $S_1$ loses one $\ell$th candidate (namely $c'$), but group $S_2$ gains two representatives instead. Hence:
    \begin{align*}
        \score_\pav((X\setminus \{c'\})\cup\{c_1, c_2\}) - \score_\pav(X) &\geq (1+\frac{1}{2}) \cdot |S_2| - \frac{|S_1|}{\ell} \\
         &= \frac{3}{2} \cdot \frac{n}{\ell+1} - \frac{3}{2} - \frac{n}{\ell+1} + \frac{1}{\ell} \\
        &> \frac{1}{2} \cdot \frac{n}{\ell+1} - \frac{3}{2} \geq 0
    \end{align*}
    The last inequality comes from the assumption that $n \geq 3\cdot (\ell+1)$. Hence, $X$ is not elected by PAV and Base EJR is violated. This completes the second part of the proof.
\end{proof}

Let us now interpret \Cref{thm:pav-and_ejr}. Intuitively, it says that PAV gives strong proportionality guarantees if all the candidates cary the same weight in the feasibility constraints. An example type of elections where PAV fails EJR is participatory budgeting~\citep{participatoryBudgeting}, where different candidates might have different costs, and so some candidates can exploit the feasibility constraints to a higher extent than the others. We will discuss in more detail such types of constraints in \Cref{sec:weighted-candidates}.

\section{Priceable Outcomes}

In this section we take a different approach to designing proportional selection rules. It is based on the idea of priceability~\cite{pet-sko:laminar,pet-pie-sha-sko:stable-priceability}, which can be intuitively described as follows: the voters are initially endowed with some fixed amount of virtual money; this money can be spent only on buying the candidates. The voters prefer to buy such candidates for which they are asked to pay least; typically these are candidates who have higher support, as for such candidates more voters are willing to participate in a purchase. The outcome consists of the purchased candidates. This approach has already proved useful in the design of selection rules for committee elections and participatory budgeting~\cite{pet-sko:laminar,pet-pie-sha-sko:stable-priceability,pet-pie-sko:c:participatory-budgeting-cardinal}. We start by describing a rule that implements this idea in a sequential manner.

\subsection{Phragm\'en's Sequential Method}\label{sec:phragmen}

In this section we define a natural extension of the Phragm\'en's Sequential Method~\cite{Phra96a,Janson16arxiv,aaai/BrillFJL17-phragmen} (a~method known from the setting of committee elections) to the model with general constraints.

\begin{definition}[Phragm\'en's Sequential Method]\label{def:phragmen}
 We start with an empty outcome $W = \emptyset$. The price for each candidate $c$ is $1$ dollar; this cost needs to be covered by the supporters of $c$ if $c$ is selected. Voters earn money continuously at a constant speed (say, $1$ dollar per second). At each time moment, when a group of supporters of some candidate $c$ has $1$ dollar in total, $c$ is purchased: we set $W \gets W \cup \{c\}$ and reset the budgets of the voters from $N(c)$ to $0$. After that, we remove from the election all the candidates $c'$ such that $W\cup\{c'\}\notin \feasibles$ and continue the procedure until all the candidates are either purchased or removed.
\end{definition}

While the definition assumes that time and money are continuous, the exact moments of purchasing next candidates can be computed in polynomial time. It is known that in the context of committee elections the Phragm\'en's sequential method fails EJR~\cite{aaai/BrillFJL17-phragmen}, but nevertheless it has very good properties pertaining to proportionality~\cite{lac-sko:multiwinner-book}. In particular, it satisfies the axiom of  Proportional Justified Representation (PJR)~\cite{pjr17}, a weaker variant of EJR. We will show that the method preserves its good properties as long as the constraints have a matroid structure.

We start by generalising the axiom of Proportional Justified Representation to the case of general constraints. It differs from EJR (\Cref{def:main-ejr}) in a way how the value of $\ell$ is defined.

\begin{definition}[Proportional Justified Representation (PJR)]\label{def:pjr}
  Given an election $E = (C, N, \feasibles, \mathcal{A})$, we say that a group of voters $S \subseteq N$ \myemph{weakly deviates} in an outcome $W\in\feasibles$ if for $\ell= |(\bigcup_{i\in S} A_i) \cap W |+ 1$ and for each set $T\subseteq W$ either there exists $X \subseteq \bigcap_{i\in S} A_i$ with $|X| \geq \ell$ such that $T \cup X \in \feasibles$, or the following inequality holds:
\begin{align}\label{eq:tpjr_condition}
\frac{|S|}{n} > \frac{\ell}{|T| + \ell} \text{.}
\end{align}
We say that a feasible outcome $W \in \feasibles$ of an election $E = (C, N, \feasibles,\mathcal{A})$ satisfies \myemph{proportional justified representation (PJR)} if no group of voters $S\subseteq N$ weakly deviates in $W$.
\hfill $\lrcorner$
\end{definition}

The axiom of Base PJR is defined analogously.

\begin{definition}[Base Proportional Justified Representation (BPJR)]\label{def:pjr}
  We say that a feasible outcome $W \in \feasibles$ of an election $E = (C, N, \feasibles,\mathcal{A})$ satisfies \myemph{base proportional justified representation (BPJR)} if for each $\ell \in \naturals$ and each group of voters $S\subseteq N$ that deserves $\ell$  candidates (according to \Cref{def:ejr}) there are at least $\ell$ candidates from $\bigcup_{i\in S} A_i$ in $W$.
\end{definition}

We will now prove that Phragm\'en's Sequential Method satisfies PJR if and only if the feasibility constraints have a matroid structure.

\begin{theorem}\label{thm:phragmen-and_pjr}
    Phragm\'en's Sequential Method satisfies PJR for elections with matroid constraints. For each non-matroid feasibility constraints there is an election where the method fails Base PJR.
\end{theorem}
\begin{proof}
    We start by proving the first part of the theorem statement. Towards a contradiction assume there is a group $S$ weakly deviating in an outcome returned by Phragm\'en's Sequential Method, $W_{\mathrm{Phr}}$. Let $\ell= |(\bigcup_{i\in S} A_i) \cap W_{\mathrm{Phr}} |+ 1$. Consider the first moment, $t$, when all the candidates from $\bigcap_{i\in S} A_i$ are either elected or removed. Note that $t \leq \nicefrac{\ell}{|S|}$. Indeed, if $t > \nicefrac{\ell}{|S|}$, then at time $\nicefrac{\ell}{|S|}$ the group $S$ would collect in total $\ell$ dollars, and would buy at least $\ell$ candidates from $\bigcup_{i\in S} A_i$ (the possibility of buying such candidates comes from the fact that there would always be a candidate from $\bigcap_{i\in S} A_i$ available for purchase).

     Let $W \subseteq W_{\mathrm{Phr}}$ denote an outcome purchased at time $t$. The voters from $N \setminus S$ could spend at most $(n-|S|)\cdot \nicefrac{\ell}{|S|} = \nicefrac{n\cdot \ell}{|S|}-\ell$ dollars on candidates from set $T = W\setminus \bigcup_{i\in S} A_i$. In particular, as the price for all the candidates is $1$, it means that:
    \begin{equation*}
        |T| \leq \frac{n\cdot \ell}{|S|}-\ell.
    \end{equation*}

    Now we need to consider two cases. First, suppose that there is no set $X\subseteq \bigcap_{i\in S} A_i$ of size $\ell$ such that $T \cup X \in \feasibles$. Then, as $S$ is weakly deviating, the following inequality holds:
    \begin{equation*}
        \frac{|S|}{n} > \frac{\ell}{|T| + \ell} \geq \frac{\ell}{\nicefrac{n\cdot \ell}{|S|}-\ell + \ell} = \frac{|S|}{n},
    \end{equation*}
    a contradiction. Hence, there exists set $X\subseteq \bigcap_{i\in S} A_i$ of size $\ell$ such that $T \cup X \in \feasibles$.

    Since we assumed that $|W \cap \bigcup_{i\in S} A_i| < \ell$, we infer that $W$ has strictly smaller size than $T\cup X$. Now we can apply the exchange property \ref{cond:exchange-property} to sets $W$ and $T\cup X$ to obtain that there exists a candidate $c\in X \setminus W$ such that $W \cup \{c\}\in \feasibles$. But then we have that $c\in \bigcap_{i\in S} A_i$ and $c$ was neither elected (since $c \notin W$) nor removed (since $W \cup \{c\}\in \feasibles$). We obtain a contradiction.

    Now we prove that if the feasibility constraints are not matroid constraints, then Phragm\'en's Sequential Method violates Base PJR. Suppose that for given constraints, there exist nonempty sets $X, Y\in \feasibles$ such that $|X| < |Y|$ and $X\cup\{c\}\notin \feasibles$ for all $c\in Y\setminus X$. Then it holds that $X\nsubseteq Y$ and so $|Y\setminus X| \geq 2$. Denote by $\ell = |Y|$.

    Consider the following construction. All the voters approve candidates from $X$. Additionally, we have a group $S$ of $\nicefrac{\ell}{\ell+1}\cdot n + 1$ voters, each of whom additionally approves $Y$.

    Let us first show that $S$ deserves $\ell$ candidates. Indeed, for $T=\emptyset$ this group can propose the feasible set $Y$ of size $\ell$ they jointly approve. Otherwise, if $|T| \geq 1$:
    \begin{align*}
    \frac{|S|}{n} = \frac{\ell}{\ell+1} + \frac{1}{n} > \frac{\ell}{\ell+1} \geq \frac{\ell}{\ell+|T|} \text{.}
    \end{align*}

    Note that Phragm\'en's Sequential Method first elects all the unanimous candidates, namely $X$. However, after that no candidates from $Y\setminus X$ can be elected. Hence, $S$ would get only at most $|X| < \ell$ representatives, which completes the proof.
\end{proof}

While Phragm\'en's Sequential Method does not satisfy (Base) EJR, it still guarantees that the voters from a group deserving $\ell$ candidates have high utility on average---in fact as high as it would be implied by Base EJR (cf.~\Cref{prop:ejr-average-sat}). This result also holds for the general class of matroid constraints. We prove it by combining the ideas from \Cref{thm:phragmen-and_pjr} and from the work of \citet{skowron:prop-degree}. This result, together with the fact that the Phragm\'en's method can be computed in polynomial time, makes the rule particularly appealing and practical.

\Cref{thm:phragmen-and_pjr} generalises the recent result by~\citet{cha-goe-pet:seq-decision-making}, proved in the context of sequential decision making.

\begin{theorem}
Let $W$ be the outcome returned by Phragm\'en's Sequential Method for an election with matroid constraints. For each group of voters $S$ that deserves $\ell$ candidates in $W$ we have:
\begin{align}\label{eq:prag-avg-sat}
\frac{1}{|S|}\sum_{i \in S} |A_i \cap W| \geq \frac{\ell-1}{2} \text{.}
\end{align}
\end{theorem}
\begin{proof}
Consider a group of voters $S\subseteq N$ deserving $\ell$ candidates. Towards a contradiction assume that Inequality \eqref{eq:prag-avg-sat} does not hold.
We first define the time $t$ as follows:
\begin{align*}
t =  \frac{\ell}{|S|} + \frac{\Delta - 1}{n} \text{,}
\end{align*}
where $\Delta$ is the smallest non-negative value such that at $t$ the voters from $S$ have at most $\Delta$ unspent dollars (if such $\Delta$ does not exist, then we simply set $\Delta = 0$). There are two possibilities:
\begin{enumerate}
\item Either at $t$ there was a purchase $\gamma$ such that before the purchase the voters from $S$ had at least $\Delta$ unspent dollars, and after the purchase they had at most $\Delta$ unspent dollars, or
\item at $t$ the voters from $S$ had at least $\Delta$ unspent dollars.
\end{enumerate}
The analysis in both cases is the same, thus without loss of generality let us assume that we are in the first case.
We will first prove that at $t$, before the purchase $\gamma$ there was a candidate from $\bigcap_{i\in S} A_i$ that was neither elected nor removed.

At $t$ the voters collected in total $tn$ dollars, and they have spent at most $tn - \Delta$. Hence, they have bought the set $W$ of at most $t n  - \Delta$ candidates.
From $W$ we remove in total $\ell - 1$ candidates, and let us call the remaining set $T$; if $W$ contained fewer than $\ell-1$ candidate, then we simply set $T = \emptyset$. We can remove these $\ell - 1$ candidates in such a way that $T \cap \bigcap_{i\in S} A_i = \emptyset$ (this follows from the fact that the Inequality \eqref{eq:prag-avg-sat} is not satisfied, and so there are fewer than $\ell-1$ candidates in $W \cap \bigcap_{i\in S} A_i$).

If there exists $X \subseteq \bigcap_{i\in S} A_i$ such that $X \cup T \in \feasibles$, then by the exchange property \ref{cond:exchange-property} applied to $X \cup T$ and $W$ we infer that there must exist a candidate $c \in  \bigcap_{i\in S} A_i$ such that $W \cup \{c\} \in \feasibles$ (which is exactly what we wanted to prove). Otherwise, since $S$ deserves $\ell$ candidates we get that:
\begin{align*}
\frac{|S|}{n} > \frac{\ell}{|T| + \ell} \text{.}
\end{align*}
In particular, this means that $T \neq \emptyset$ and so $T = |W| - \ell + 1$. Consequently, we get that:
\begin{align*}
\frac{|S|}{n} > \frac{\ell}{|W| + 1} \text{.}
\end{align*}
From that we get that
\begin{align*}
|W| > n \cdot \left(\frac{\ell}{|S|} - \frac{1}{n} \right) = nt - \Delta \text{.}
\end{align*}
This leads to a contradiction.

Thus, at time $t$ there is a candidate from $\bigcap_{i\in S} A_i$ available for purchase. We can now use exactly the same reasoning as the one provided in the work of \citet{skowron:prop-degree}. There, using an argument involving potential functions, it was implicitly proved that at each time moment, as long as some candidate from $\bigcap_{i\in S} A_i$ can be purchased, the voters from $S$ pay on average at most $\nicefrac{2}{|S|}$ per approved candidate. Thus, the average payment per approved candidate until time $t$ was no greater than $\nicefrac{2}{|S|}$.

At time $t$ the voters from $S$ spent at least $t|S| -\Delta$ dollars in total. Let us assess this value:
\begin{align*}
t|S| -\Delta =  \left(\frac{\ell}{|S|} + \frac{\Delta - 1}{n}\right) |S| - \Delta = \ell + \frac{|S|}{n} \cdot (\Delta - 1)  - \Delta \geq  \ell - 1 \text{.}
\end{align*}
The last inequality follows from the fact that $\Delta \leq 1$ (otherwise, they money held by $S$ would be used earlier to buy a candidate $\bigcap_{i\in S} A_i$).

Consequently, the voters from $S$ approve on average at least the following number of candidates:
\begin{align*}
\frac{1}{|S|} \cdot \frac{\ell-1}{\frac{2}{|S|}} = \frac{\ell-1}{2} \text{.}
\end{align*}
This completes the proof.
\end{proof}

So far our results applied to matroid constraints only. Interestingly, if the constraints are non-matroid, then Phragm\'en's Sequential Method still can be successfully applied---in the most general case it provides approximate proportionality guarantees (cf.~\Cref{thm:phragmen-and_pjr_weighted}).

\subsection{Stable Priceable Outcomes}\label{sec:stable-priceable}

In the previous section, we have described the sequential process of buying candidates. Another approach would be to define the outcomes as an equilibrium in a certain market. This idea has been proposed by~\citet{pet-pie-sha-sko:stable-priceability}, who introduced the concept of \myemph{stable priceability}, inspired by the classic concept of Lindahl's equilibrium~\cite{foley:lindahl}. In this section we adapt the concept to the setting with general constraints. It requires introducing a few additional elements that relate candidate prices and feasibility constraints.


Let us recall the definition of \myemph{stable priceability}~\cite{pet-pie-sha-sko:stable-priceability}. Let $\pi_c$  denote the price of a candidate $c$ and $\pi_W=\sum_{c\in W}\pi_c$. Given a voter $i \in N$ the payment function $p_i \colon C \to \reals$ specifies how much the voter pays for the particular candidates. We require that $p_{i}(c) \geq 0$ for each $c \in C$ and that $\sum_{c \in C} p_i(c) = 1$; intuitively, this means each voter has the total budget of one unit. We say that an outcome $W$ is stable-priceable if there exists a collection of candidate prices $\{\pi_c\}_{c \in C}$ and payment functions $\{p_i\}_{i \in N}$ such that the following conditions hold:
\begin{enumerate}[label=(SP\arabic*)]
\item The voters pay only for the selected candidates, i.e., $p_i(c) = 0$ for each $i \in N$ and $c \notin W$.
\item \label{sp:paying-for-selected} The total payment for each selected candidate $c \in W$ must equal its price, $\sum_{i \in N}p_i(c) = \pi_c$.
\item \label{sp:condition} For each not-selected candidate $c \notin W$ we have:
\begin{align*}
\sum_{i \in N(c)} \max(r_i, \max_{c' \in W} p_i(c')) \leq \pi_c \quad \text{where} \quad r_i = 1 - \sum_{c' \in W}p_i(c').
\end{align*}
This condition can be intuitively explained as follows. Each voter is primarily interested in buying as many approved candidates as possible. Secondarily, the voter is interested in spending as little money as possible. Thus, each voter is willing to pay for $c$ either her all remaining money $r_i$ or to stop paying for some already selected candidate $c'$ and to pay the same or a lower amount for $c$ instead. If the supporters of candidate $c$ can pay its price this way, then the payment functions are not stable.
\item \label{item:producer-stability} The outcome $W$ maximizes the total price:
\begin{align*}
W \in \argmax_{W' \in \feasibles} \sum_{c \in W'} \pi_c.
\end{align*}
\end{enumerate}
The last condition is new to this paper, and it corresponds to the concept of producer-stability from the economic literature on markets with public goods~\cite{foley:lindahl}. In the original definition of~\citet{pet-pie-sha-sko:stable-priceability} the prices for all the candidates were required to be equal and hence, instead of the last condition, only exhaustiveness was required. In \Cref{sec:weighted-candidates} we show that some of our results also hold if we replace condition \ref{item:producer-stability} with the condition of exhaustiveness.
\begin{enumerate}[]\phantomsection
  \item[(SP4$^{*}$)] \label{item:producer-stability-exh} The outcome $W$ is exhaustive, i.e., for each $c \notin W$ we have $W \cup \{c\} \notin \feasibles$.
\end{enumerate}

Stable-priceable outcomes might not exist, but if they do, they have very good fairness-related properties.

\begin{theorem}\label{thm:sp-ejr}
For elections with matroid constraints all stable-priceable outcomes satisfy EJR.
For elections with general constraints each stable-priceable outcome such that all candidates $c$ have a common price $\pi_c = \pi$ satisfies EJR.
\end{theorem}
\begin{proof}
Let $W$ be a stable-priceable outcome. Towards a contradiction assume that there is a group of voters $S$ that deviates in $W$. Let $\ell=\max_{i \in S} u_i(W) + 1$.
We will first prove that there exists a not-selected candidate $a \in \bigcap_{i \in S} A_i \setminus W$, the price of which satisfies the following inequality:
\begin{align*}
\pi_a < \frac{|S|}{\ell} \text{.}
\end{align*}
Let us start with the case of matroid constraints. We first consider the set $W'$ of candidates from $W$ that can be feasibly-exchanged with the candidates from $\bigcap_{i \in S} A_i \setminus W$, that is:
\begin{align*}
W' = \left\{c \in W \colon \text{there exists~}c'\in \bigcap_{i \in S} A_i \setminus W \text{~such that~} W\setminus\{c\}\cup\{c'\} \in \feasibles \right\} \text{.}
\end{align*}
Since the feasibility constraints have a matroid structure, from \Cref{lem:one-swap} we know that for each $c'\in \bigcap_{i \in S} A_i \setminus W$ it holds that $W' \cup \{c'\} \notin \feasibles$. Next, from $W'$ we remove all the candidates in $\bigcap_{i \in S} A_i$. Clearly, we removed at most $\ell-1$ candidates; if we removed strictly less than $(\ell - 1)$ candidates, then we additionally remove some arbitrary candidates so that we removed in total $(\ell-1)$ candidates. Let us denote the resulting set as $T$. Note that for each $X \subseteq \bigcap_{i \in S} A_i$ with $|X| = \ell$ it holds that $T \cup X \notin \feasibles$. This follows directly from the  exchange property \ref{cond:exchange-property} applied to $X \cup T$ and $W'$. Since $S$ deviates in $W$, we get that:
\begin{align*}
\frac{|S|}{n} > \frac{\ell}{|T| + \ell} = \frac{\ell}{|W'| + 1} \text{.}
\end{align*}
After reformulating, we get that:
\begin{align*}
|W'| > n \cdot \frac{\ell}{|S|} - 1\text{.}
\end{align*}

Let $a$ be the cheapest candidate in $\bigcap_{i \in S} A_i \setminus W$.
Condition \ref{item:producer-stability} in the definition of stable-priceability implies that for each candidate $c \in W'$ we have $\pi_c \geq \pi_a$. Consequently:
\begin{align*}
\sum_{c \in W'}\pi_c > \pi_a \cdot \left(n \cdot \frac{\ell}{|S|} - 1\right)\text{.}
\end{align*}
Since $\sum_{c \in W'}\pi_c \leq n - \sum_{i \in S}r_i$ (this follows from condition \ref{sp:paying-for-selected}) we get that:
\begin{align*}
n - \sum_{i \in S}r_i > \pi_a \cdot \left(n \cdot \frac{\ell}{|S|} - 1\right) \text{.}
\end{align*}
By condition \ref{sp:condition} we also get that
\begin{align*}
\sum_{i \in S}r_i \leq \pi_a \text{.}
\end{align*}
By combining the two above inequalities we get that:
\begin{align*}
n > \pi_a n \cdot \frac{\ell}{|S|} \text{.}
\end{align*}
This is equivalent to:
\begin{align*}
\pi_a < \frac{|S|}{\ell} \text{.}
\end{align*}

Now, consider the case where the feasibility constraints are arbitrary, but all the prices are the same, that is for all candidates $c$ we have $\pi_c = \pi$. Now, we proceed as follows.
From $W$ we remove all the candidates from $\bigcap_{i \in S} A_i$ and some arbitrary additional candidates so that in total we removed $\ell-1$ candidates. Let us denote the resulting set by $T$.

Condition \hyperref[item:producer-stability-exh]{(SP4$^{*}$)} in the definition of stable-priceability implies that for any set $X \subseteq \bigcap_{i \in S} A_i$ with $|X| = \ell$, $T \cup X \notin \feasibles$. Since $S$ deviates in $W$ we get that:
\begin{align*}
\frac{|S|}{n} > \frac{\ell}{|T| + \ell} = \frac{\ell}{|W| + 1} \text{.}
\end{align*}
After reformulating, we get that:
\begin{align*}
|W| > n \cdot \frac{\ell}{|S|} - 1\text{.}
\end{align*}
Since $\pi W \leq n - \sum_{i \in S}r_i$ (this follows from condition \ref{sp:paying-for-selected}) we get that:
\begin{align*}
n - \sum_{i \in S}r_i > \pi n \cdot \frac{\ell}{|S|} - \pi \text{.}
\end{align*}
The remaining part of the proof follows exactly the same way as in the case of matroid constraints. Thus, in either case, we get that there is a candidate $a \in \bigcap_{i \in S} A_i \setminus W$ such that:
\begin{align*}
\pi_a < \frac{|S|}{\ell} \text{.}
\end{align*}

Since each voter from $i \in S$ approves strictly fewer candidates than $\ell$ in $W$ we infer that:
\begin{align*}
 \max(r_i, \max_{c \in W} p_i(c)) \geq \frac{1}{\ell} \text{.}
\end{align*}
Thus, from condition \ref{sp:condition} in the definition of stable-priceability we get for each $a\notin W$ that:
\begin{align*}
\frac{1}{\ell}  \cdot |S| \leq \pi_a \text{.}
\end{align*}
This gives a contradiction and completes the proof.
\end{proof}

The condition for stable-priceability can be easily written as an Integer Linear Program with the number of integer variables bounded by the number of candidates. Further, the ILP can be naturally relaxed so that it finds outcomes that are ``closest to'' stable-priceable. This makes the approach applicable to elections of moderate size.

\section{Extensions of the Model}\label{sec:extensions}

In this section we consider two extensions of our model. We first discuss the case where the preferences of the voters are expressed as general monotone set functions. Second, we discuss certain limitation of our concepts in the case when the candidates cary different weights in the feasibility constraints; we explain how to adapt our concepts to this case.

\subsection{General Monotone Utility Functions}\label{sec:geenral-utility-functions}

In this section, we formulate a stronger version of \Cref{def:ejr} that still is satisfiable. This definition applies much beyond approval ballots. Here, we assume that for each voter $i$ we have a utility function $u_i\colon 2^C \to \reals$ that for each subset of candidates returns a real value. Intuitively, $u_i(W)$ quantifies the level of satisfaction of voter $i$ provided $W$ is selected. We only assume that $u_i$ is monotone, that is for all $X \subseteq Y \subseteq C$ it holds that $u_{i}(Y) \geq u_i(X)$.

Now, we can formulate the axiom of fully justified representation, which generalizes the respective axiom from the literature on committee elections~\citep{pet-pie-sko:c:participatory-budgeting-cardinal}.

\begin{definition}[Base Fully Justified Representation (BFJR)]\label{def:bfjr}
  Given an election $E = (C, N, \feasibles,\mathcal{A})$ we say that a group of voters $S \subseteq N$ is $(\alpha, \beta)$-cohesive, $\alpha, \beta \geq 0$, if for each feasible set $T \in \feasibles$ either there exists $X \subseteq C$ with $|X| = \alpha$ and with $u_i(X) \geq \beta$ for each $i \in S$ such that $T \cup X \in \feasibles$\footnote{Note that if $\alpha=0$ then $T\cup\emptyset\in\feasibles$ and in that case we do not need to consider the or condition where we could potentially divide by 0 when $T=\emptyset$.}, or the following condition is satisfied:
\begin{align*}
\frac{|S|}{n} > \frac{\alpha}{|T| + \alpha} \text{.}
\end{align*}
We say that a feasible outcome $W \in \feasibles$ of an election $E = (C, N, \feasibles,\mathcal{A})$ satisfies \myemph{base fully justified representation (BFJR)} if for each $\alpha, \beta \in \reals$ and each  $(\alpha, \beta)$-cohesive group of voters $S\subseteq{}N$ there exists a voter $i \in S$ such that $u_i(W) \geq \beta$. A selection rule $\calR$ \emph{satisfies} BFJR if for each election $E$ it returns outcomes satisfying BFJR.  \hfill $\lrcorner$
\end{definition}

\Cref{def:bfjr} is strictly stronger than \Cref{def:ejr}: indeed we obtain the definition of BEJR if we additionally require that $\alpha = \beta$. Intuitively, in the definition of BFJR a group of voters $S$ deserves the utility of $\beta$ if for each $T$ they can find a set $X$ of (not too large) size $\alpha$ on which they agree that it has the value of at least $\beta$. In the definition of BEJR, the voters from $S$ must have a stronger agreement; they all must unanimously support every candidate from $X$. 
This definition is still satisfiable.

\begin{theorem}\label{thm:fjr:existence}
For each election with monotone utilities there exists an outcome satisfying base fully justified representation.
\end{theorem}
\begin{proof}
  Given an election $E=(C, N, \feasibles, \mathcal{A})$ we first define the procedure of partitioning voters. In each round $r$ we search for the largest value $\beta_r \geq 0$ for which there exists an $(\alpha_r, \beta_r)$-cohesive group, $\alpha_r \geq 0$; if there are ties, we first prefer a cohesive group with a smaller value of $\alpha_r$. We pick one such a group, call it $S_r$, and remove the voters from $S_r$ from further consideration. We repeat the procedure until all the voters are removed (note that every non-empty group of voters is $(0, 0)$-cohesive, and so the procedure will stop). Thus, we partitioned the set of voters into disjoint sets $S_1, S_2, \ldots, S_{p-1}, S_p$.

We will show now that for each group of voters $S_r$, $r \in [p]$, we can select a set of candidates $W_r$ with $|W_r| \leq \alpha_r$ such that
\begin{inparaenum}[(1)]
\item each set $W_r$ gives the voters from $S_r$ the utility of at least $\beta_r$ (that is $u_i(W_r) \geq \beta_r$ for each $i \in S_r$), and
\item the set $W_1 \cup W_2 \cup \ldots \cup W_p$ is feasible.
\end{inparaenum}

We first fix the number of voters $n$. We will show the above statement by the induction on the number of active voters; the voter is active if it assigns a positive utility to some subset of candidates. Clearly, if all the voters are inactive then the inductive hypothesis holds, which is witnessed by an empty subset of candidates. Now, assume that the hypothesis holds if the number of active voters is strictly lower than $n'$. We will show that it holds also for $n'$.

Consider any set $S_r\in \{S_1, \ldots, S_p\}$ and consider a modified election $E'$ in which we replace each voter from $S_r$ with an inactive voter. Note that except for $S_r$ the partitioning algorithm would return the same groups $S_1, S_2, \ldots, S_{r-1}, S_{r+1}, \ldots, S_p$ as for $E$.
From our inductive assumption, there exists sets $W'=W_1 \cup W_2 \cup \ldots\cup W_{r-1} \cup W_{r+1} \cup \ldots \cup W_p$ such that all voters from $S_i$ get at least the utility of $\beta_i$ from $W_i$ for $i\in [p]\setminus \{r\}$.

For $T=W'$, since $S_r$ is $(\alpha_r, \beta_r)$-cohesive, we know that there exists $X\in \feasibles$ with $|X| = \alpha_r$ that gives $S_r$ the utility of $\beta_r$ such that $X \cup T \in \feasibles$ or
\begin{align*}
\frac{|S_{r}|}{n} > \frac{\alpha_r}{|T| + \alpha_r} \text{.}
\end{align*}
If $X \cup T \in \feasibles$ then we set $W_r = X$; we additionally give empty sets to inactive voters, and we are done. Otherwise, we have that:
\begin{align*}
\frac{|S_{r}|}{n} > \frac{\alpha_r}{|T| + \alpha_r} \geq \frac{\alpha_r}{\alpha_1 + \ldots + \alpha_p}\text{.}
\end{align*}
We repeat this reasoning for each $r \in [p]$, and get that unless we are done:
\begin{equation*}
    1 = \frac{|S_1| + \ldots + |S_{p}|}{n}  > \frac{\alpha_1 + \ldots + \alpha_p}{\alpha_1 + \ldots + \alpha_p} = 1,
\end{equation*}
a contradiction. Hence, there exists a sequence $W_1 \cup \ldots \cup W_p$ that satisfies our condition.

Now, consider an election $E=(C, N, \feasibles, \mathcal{A})$, and let  $W_1, W_2, \ldots, W_p$ be constructed as above. We take $W = W_1 \cup W_2 \cup \ldots \cup W_p$.
It remains to prove that $W$ satisfies BFJR. For the sake of contradiction, assume that there exists an $(\alpha, \beta)$-cohesive group $S$ such that for each voter $i \in S$ we have $u_i(W) < \beta$. Consider the first step in the process of partitioning the voters, when some voter $i \in S$ was deleted. It was deleted as a part of some $(\alpha_r, \beta_r)$-cohesive group $S_r$. Since during the partitioning we always pick the group with the highest $\beta_r$ first, we know that $\beta_r \geq \beta$. From the construction of $W$ we know that $u_i(W) \geq u_i(W_r) \geq \beta_r \geq \beta$. This proves a contradiction, and completes the proof.
\end{proof}

Now let us investigate the relation between BFJR and stable priceability. In the committee setting, stable priceability implies the core \cite{pet-pie-sha-sko:stable-priceability} which is a stronger axiom than BFJR. However, \Cref{thm:stable-priceability-fjr} shows that it is no longer the case in our general model, even if we assume matroid constraints.

\begin{proposition}\label{thm:stable-priceability-fjr}
Stable Priceability with different prices does not imply BFJR for approval committee elections with disjoint attributes even if we assume matroid constraints.
\end{proposition}
\begin{proof}
  Consider an instance of approval committee elections with disjoint attributes, $E = (C, N, \feasibles,\mathcal{A})$ such that $C=C_1 \sqcup C_2$ (candidates are split into two separate groups) and $|C_1| \geq 41$, $|C_2| \geq 50$. Feasible sets contain at most $40$ candidates from the first group and at most $40$ from the second group. There are three disjoint groups of voters: group $V_1$ of $\nicefrac{3n}{5}$ approving all the candidates from $C_1$, group $V_2$ of $\nicefrac{n}{3}$ of voters approving all the candidates from $C_2$, and group $S$ of $\nicefrac{n}{15}$ (assume $\nicefrac{n}{15}$ is an even integer) voters approving some $5$ candidates $A=\{a_1, \ldots, a_5\}$ from $C_1$. Besides, half of the voters from $S$ approve some $5$ candidates $B=\{b_1, \ldots, b_5\}$ from $C_2$ and the other half of the voters approve different $5$ candidates $E=\{e_1, \ldots, e_5\}$ from $C_2$.

    Consider now an outcome $W$ containing $4$ candidates from $A\setminus \{a_5\}$, $36$ other candidates from $C_1\setminus \{a_5\}$, and $40$ candidates from $C_2\setminus B \setminus E$. We will first show that this outcome is stable priceable. Let the price for the candidates from $C_1$ be $\pi_1=\nicefrac{n}{60}$ and the price for the candidates from $C_2$ be $\pi_2=\nicefrac{n}{120}$. It is clear that with such prices, outcome $W$ (and every other outcome with $40$ candidates from $C_1$ and $40$ candidates from $C_2$) satisfies \ref{item:producer-stability}.

    Voters from $S$ spend all their money on candidates from $A\setminus \{a_5\}$ (indeed: $4\cdot \nicefrac{n}{60}=\nicefrac{n}{15}=|S|$, the remaining voters spend their money on their approved candidates from $W$ in any way so that each elected candidate is paid her price (it is possible, since $36\cdot \nicefrac{n}{60} = \nicefrac{3n}{5} = |V_1|$ and $40\cdot \nicefrac{n}{120}=\nicefrac{n}{3}=|V_2|$). Now we can see that $W$ is stable: indeed, voters from neither $V_1$ nor $V_2$ have no possibility to improve their satisfaction from the committee. Voters from $S$ do not have enough money to buy the fifth candidate from $A$. Besides, even after resigning from paying for $A\setminus \{a_5\}$, they do not have enough money to improve their satisfaction by buying candidates from $B\cup E$.

    We will now show that in any committee satisfying BFJR, some member of group $S$ should get at least $5$ representatives.

    Consider any committee $T\in \feasibles$. If $T$ contains less than $36$ candidates from $C_1$ then clearly group $S$ can propose a committee $X=A$ satisfying $X\cup T\in \feasibles$. If $T$ contains some $35+x$ candidates from $C_1$ for $x\in [5]$ but less than $41-2x$ candidates from $C_2$ then voters from $S$ can propose a committee $X$ containing $5-x$ candidates from $A$, $x$ candidates from $B$ and $x$ candidates from $E$ satisfying $X\cup T\in \feasibles$. Consider now a committee $T$ containing some $35+x$ candidates from $C_1$ (for $x\in [5]$) and at least $41-2x$ candidates from $C_2$ (hence $|T| \geq 76-x \geq 71$). Now consider $X=A$. We have that:
    \begin{equation*}
        \frac{|X|}{|T|+|X|}\cdot n \leq \frac{5}{71+5}\cdot n < \frac{n}{15} = |S|,    \end{equation*}
    which shows that $S$ indeed should get $5$ representatives in any committee satisfying BFJR which completes the proof.
\end{proof}

\subsection{Weighted Candidates}\label{sec:weighted-candidates}

It is worth noting that our definitions and the analysis so far implicitly assumed that all the candidates are treated equally, irrespectively of their impact on the feasibility constraints. Specifically, in Inequality~\eqref{eq:ejr_condition} in \Cref{def:ejr} we were only concerned with the number of candidates in the set $T$; however, some of these candidates can restrict the feasible solutions much more than the others. The classic model where this is the case is the one of participatory budgeting (PB)~\cite{pet-pie-sko:c:participatory-budgeting-cardinal}---there, the candidates have weights, and there is a single constraint specifying that the total weight of the selected candidates is lower than or equal to the given budget value. In such cases it might be justified to include the weights of the candidates in the definitions of the axioms, as it is done in the work of~\citet{pet-pie-sko:c:participatory-budgeting-cardinal}.

In this section we are considering the following addition to the original model. For each candidate $c \in C$ assume we are given a weight $\weight(c) \in \reals_{+}$; for a subset of candidates $W \subseteq C$ we let $\weight(W) = \sum_{c \in W} \weight(c)$. Intuitively, the weights of the candidates should in some way correspond to the feasibility constraints, however this is not formally required. In this case, we write the definition of Base Extended Justified Representation as follows.

\begin{definition}[Weighted Base Extended Justified Representation ($\weight$-BEJR)]\label{def:wbejr}
  Given an election $E = (C, N, \feasibles, \mathcal{A})$ we say that a group of voters $S \subseteq N$ is \myemph{$(\alpha, \beta)$-strongly cohesive}, $\alpha>0$, $\beta \geq 0$, if for each feasible set $T \in \feasibles$ either there exists $X \subseteq \bigcap_{i \in S} A_i$ with $\weight(X) \leq \alpha$ and $|X| \geq \beta$ such that $T \cup X \in \feasibles$, or the following condition is satisfied:
\begin{align*}
 \frac{|S|}{n} > \frac{\alpha}{\weight(T) + \alpha} \text{.}
\end{align*}
We say that a feasible outcome $W \in \feasibles$ of an election $E = (C, N, \feasibles, \mathcal{A})$ satisfies \myemph{weighted base extended justified representation ($\weight$-BEJR)} if for each $\alpha>0, \beta\geq 0$ and each $(\alpha, \beta)$-strongly cohesive group of voters $S\subseteq{}N$ there exists a voter $i \in S$ such that $|W \cap A_i| \geq \beta$.
\end{definition}

If we additionally restrain $T$ to be a subset of $W$, then we say that $S$ is $(\alpha, \beta)$-strongly cohesive in $W$.
Analogously, we extend the definition of EJR to the case of weighted candidates. The easiest formulation is obtained by modifying \Cref{def:wbejr} so that we consider strongly cohesive groups in $W$. Below we also provide an alternative equivalent formulation.

\begin{definition}[Weighted Extended Justified Representation ($\weight$-EJR)]\label{def:wejr}
  Given an election $E = (C, N, \feasibles, \mathcal{A})$ and a real $\alpha>0$, we say that a group of voters $S \subseteq N$ \myemph{$\alpha$-deviates} in some $W\in\feasibles$ if for $\ell=\max_{i \in S} u_i(W) + 1$ and for each set $T\subseteq W$ either there exists $X \subseteq \bigcap_{i\in S} A_i$ with $\weight(X) \leq \alpha$ and $|X|\geq \ell$ such that $T \cup X \in \feasibles$, or the following inequality holds:
\begin{align}\label{eq:wtejr_condition}
  \frac{|S|}{n} > \frac{\alpha}{\weight(T) + \alpha} \text{.}
\end{align}
We say that a feasible outcome $W \in \feasibles$ of an election $E = (C, N, \feasibles,\mathcal{A})$ satisfies \myemph{weighted extended justified representation ($\weight$-EJR)} if no group of voters $S\subseteq N$ $\alpha$-deviates in $W$ for any $\alpha>0$.
\hfill $\lrcorner$
\end{definition}

We analogously extend the definitions of fully justified representation (FJR) and proportional justified representation (PJR) to the case of weighted candidates.
We will also say that a group of voters $S$ \myemph{strongly deserves} $\beta$ candidates (in $W$) if this group is $(\alpha, \beta)$-strongly cohesive (in $W$) for some $\alpha>0$.

It is known that in the model with weights Proportional Approval Voting (PAV) fails EJR and PJR, even in approximation~\cite{pet-pie-sko:c:participatory-budgeting-cardinal}. The approach based on priceability, on the other hand, provides more positive results. For Phragm\'en's Sequential Method it suffices to assume that the costs of the candidates that need to be paid by the voters equal to their weights. We will show that while, such defined rule fails all our axioms (PJR, and consequently EJR, and FJR), it provides certain approximate guarantees. The method provides particularly strong guarantees to small cohesive groups. For instance, a cohesive group of the size of 10\% of the population, is guaranteed roughly 0.9 of PJR. It is worth noting that the approximation result works also for non-matroid constraints.
The following theorem gives the guarantee already for the $\weight$-BEJR.

\begin{theorem}\label{thm:phragmen-and_pjr_weighted}
  For weighted candidates Phragm\'en's Sequential Method selects an outcome $W$ such that for each a group of voters $S \subseteq N$ strongly deserving $\beta$ candidates in $W$ we have:
    \begin{align*}
    \left|W \cap \left(\bigcup_{i\in S} A_i\right) \right| \geq \left\lfloor \beta \cdot \frac{n- |S|}{n} \right\rfloor\text{.}
    \end{align*}
\end{theorem}
\begin{proof}
  Consider an $(\alpha, \beta)$-strongly cohesive group of voters $S\subseteq N$. Towards a contradiction, assume that Phragm\'en's Sequential Method selects fewer than $\left\lfloor \beta \cdot \frac{n- |S|}{n} \right\rfloor$ candidates from $\bigcup_{i\in S}A_i$. Note that during the execution of the Phragm\'en's method, when a candidate $c \in \bigcap_{i\in S}A_i$ is not removed nor selected, then the voters from $S$ will pay no more than $\weight(c)$ in total for any candidate (as otherwise they would prefer to buy $c$). Let $t$ be the first moment when at least one candidate has been removed from each subset $A \subseteq \bigcap_{i\in S}A_i$ with $|A| \geq \beta$ and $\weight(A) \leq \alpha$. Note that $t \leq \nicefrac{n- |S|}{n} \cdot \nicefrac{\alpha}{|S|}$. Indeed, if $t > \nicefrac{n- |S|}{n} \cdot \nicefrac{\alpha}{|S|}$, then at time $\nicefrac{n- |S|}{n} \cdot \nicefrac{\alpha}{|S|}$ the group $S$ would collect in total $\alpha \cdot \nicefrac{n- |S|}{n}$ dollars. Further, at this time moment there would be a set  $A \subseteq \bigcap_{i\in S}A_i$ with $|A| \geq \beta$ and $\weight(A) \leq \alpha$ such that all candidates from $A$ would be either bought, or available for being bought. Thus, the voters from $S$ would buy at least $\left\lfloor \beta \cdot \frac{n- |S|}{n} \right\rfloor$ candidates.

     Let $W$ denote an outcome purchased at time $t$. Since the voters could spend at most $nt$ dollars on candidates from $W$, we get that:
    \begin{equation*}
        \weight(W) \leq n\cdot \frac{\alpha}{|S|} \cdot \frac{n- |S|}{n}\text{.}
    \end{equation*}

  Since $S$ is $(\alpha, \beta)$-strongly cohesive, and as the existence of $X$ such that $X \subseteq \bigcap_{i\in S}A_i$ with $|X| \geq \beta$, $\weight(X) \leq \alpha$, and $W\cup X\in\feasibles$ is not possible at time moment $t$, the following inequality holds:
    \begin{equation*}
        \frac{|S|}{n} > \frac{\alpha}{\weight(W) + \alpha} \geq \frac{\alpha}{n\cdot \frac{\alpha}{|S|} \cdot \frac{n- |S|}{n}+ \alpha} = \frac{|S|}{n},
    \end{equation*}
    a contradiction. This completes the proof.
\end{proof}

\begin{proposition}
  For weighted candidates Phragm\'en's Sequential Method may fail $\weight$-BPJR.
\end{proposition}
\begin{proof}
  Consider the following election. The candidates are divided into 100 disjoint groups, $C_1, C_2, \ldots, C_{100}$. In each group $C_i$ there are four candidates: $a_i$ with the cost equal to $2 + \epsilon$ for $\epsilon>0$ and $b_i$, $c_i$, $d_i$ with the costs equal to 1, each. The feasibility constraints are the following: for each group $C_i$ the total cost of the candidates selected from $C_i$ cannot exceed 3. Thus, from each group $C_i$ we can select either $a_i$ or $b_i$, $c_i$, and $d_i$. Let $A,B,C,D$ be a set of all $a_i, b_i, c_i, d_i$ candidates, respectively.

Consider a group $S$ consisting of $(50-\epsilon)$\% voters. The voters from $S$ all approve $B\cup C\cup D$. Additionally, all the voters (including those from $S$) approve $A$. It is straightforward to check that Phragm\'en's Sequential Method will select $A$ only; in total 100 candidates. However, we will show that the group $S$ is $(120, 120)$-strongly cohesive.

Indeed, if $T$ contains possible candidates (either $a_i$ or all of $b_i,c_i,d_i$) for at most 60 $i\in[100]$, then the group $S$ can easily point 120 candidates that together with $T$ make a feasible set. Otherwise, $\weight(T) > 122$, and so for sufficiently small $\epsilon$ it holds that:
\begin{align*}
\frac{|S|}{n} > \frac{120}{\weight(T) + 120} = \frac{120}{122 + 120} \text{.}
\end{align*}
Thus, the group $S$ should approve in total at least 120 candidates. Hence, base PJR is failed.
\end{proof}

For stable-priceability (\Cref{sec:stable-priceable}) a few more adaptations need to be made:
\begin{enumerate}
\item The assumption that the prices of the candidates are all equal would now correspond to the assumption that the prices are proportional to the candidates' costs.
\begin{align}
  \frac{\pi_c}{\weight(c)} = \frac{\pi_{c'}}{\weight(c')} \text{.} \label{eq:sp:prices}
\end{align}
\item Further, condition \ref{item:producer-stability} might be too restrictive (especially if the prices of the candidates are fixed). Indeed, such condition could itself imply a unique outcome, independently of the voters' preferences. Thus, we propose to replace it with the condition of exhaustiveness \hyperref[item:producer-stability-exh]{(SP4$^{*}$)}.
\end{enumerate}

\begin{theorem}\label{thm:sp-ejr_weighted}
 Let $W$ be a stable-priceable outcome for weighted candidates, assuming prices are proportional to the candidates' costs and the exhaustiveness of $W$. Then, for each a group of voters $S \subseteq N$ deserving $\beta$ candidates there exists a voter $i \in S$ with:
    \begin{align*}
    \left|W \cap A_i\right| \geq \left\lfloor \beta \cdot \frac{n- |S|}{n} \right\rfloor\text{.}
    \end{align*}
\end{theorem}
\begin{proof}
  Let $W$ be a stable-priceable outcome. Consider an $(\alpha, \beta)$-strongly cohesive group of voters $S\subseteq N$, and towards a contradiction, assume that each voter from $S$ approves fewer than $\left\lfloor \beta \cdot \frac{n- |S|}{n} \right\rfloor$ candidates from $W$. If $\beta=0$ or $S=N$ then the statement is trivially true, so we can assume that $\beta\geq 1$ and $S\subset N$.

  Since $W$ is exhaustive, and since $S$ is $(\alpha, \beta)$-strongly cohesive for some $\alpha>0$, $\beta\geq 0$, we get that:
\begin{align*}
\frac{|S|}{n} > \frac{\alpha}{\weight(W) + \alpha} \text{.}
\end{align*}
After reformulating, we get that:
\begin{align*}
\weight(W) > \alpha \cdot \frac{n - |S|}{|S|} \text{.}
\end{align*}
For each $c\in C$ by \eqref{eq:sp:prices} we have:
\begin{align*}
\frac{\weight(c)}{\pi_c} \geq \frac{\weight(W)} {\pi_W} \text{.}
\end{align*}
Thus, we get that:
\begin{align*}
\frac{\weight(c)}{\pi_c} > \frac{\alpha}{\pi_W} \cdot \frac{n - |S|}{|S|} \geq \frac{\alpha}{n} \cdot \frac{n - |S|}{|S|} \text{.}
\end{align*}
Further, since there exists $X \subseteq \bigcap_{i \in S} A_i$ with $|X| \geq \beta$, $\weight(X) \leq \alpha$, we get that:
\begin{align*}
\alpha \geq \sum_{c \in X}\weight(c) > \sum_{c \in X}\pi_c \cdot \frac{\alpha}{n} \cdot \frac{n - |S|}{|S|} = \pi_X \cdot \frac{\alpha}{n} \cdot \frac{n - |S|}{|S|} \text{.}
\end{align*}
This is equivalent to:
\begin{align}\label{eq:sp_prices_upper_bound}
\pi_X < \frac{n|S|}{n - |S|} \text{.}
\end{align}

Let $z = |X \cap W|$ and let $\pi_z$ denote the total price of the candidates from $X \cap W$ paid by the voters from $S$.
Let $\zeta_i$ denote the amount of money that voter $i \in S$ has excluding the amount of money that the voter paid for the candidates from $X \cap W$. We have:
\begin{align*}
\sum_{i \in S} \zeta_i = |S| - \pi_z
\end{align*}
By our assumption each voter from $S$ approves strictly fewer than $\beta \cdot \frac{n- |S|}{n}$ candidates. In particular, each such a voter approves at most $\beta \cdot \frac{n- |S|}{n} - z - 1$ candidates from $W \setminus X$. 
This means that for each $i\in S$, there exists a candidate $c'\in W\setminus X$ for which $p_i(c')\ge \frac{\zeta_i}{\beta \cdot \frac{n- |S|}{n} - z}$ or $r_i\ge \frac{\zeta_i}{\beta \cdot \frac{n- |S|}{n} - z}$.
Thus, from condition \ref{sp:condition} in the definition of stable-priceability we get that for each $c \in X \setminus W$ 
\begin{align}\label{eq:sp_prices_lower_bound1}
\pi_c \geq \frac{|S| - \pi_z}{\beta \cdot \frac{n- |S|}{n} - z} \text{.}
\end{align}
As there is $c'\in X\cap W$ such that $\pi_{c'}\ge \frac{\pi_z}{z}$, using \ref{sp:condition}, we get
that for each $c \in X \setminus W$ it holds that:
\begin{align}\label{eq:sp_prices_lower_bound2}
\pi_c \geq \frac{\pi_z}{z} \text{.}
\end{align}
Now, let us consider two cases. First, assume that $\frac{\beta \pi_z}{z} \geq \frac{n|S|}{n - |S|}$. In this case, we use \eqref{eq:sp_prices_lower_bound2} and get that:
\begin{align*}
\sum_{c \in X} \pi_c  &= \sum_{c \in X \setminus W} \pi_c + \pi_z \geq (\beta - z) \cdot  \frac{\pi_z}{z} + \pi_z \\
&=  \frac{\beta \pi_z - z\pi_z + z \pi_z}{z} = \frac{\beta \pi_z}{z} \geq \frac{n|S|}{n - |S|}\text{.}
\end{align*}
This gives a contradiction with \eqref{eq:sp_prices_upper_bound}. Now assume the opposite case, that is: $\frac{\beta \pi_z}{z} < \frac{n|S|}{n - |S|}$. 
Here we use \eqref{eq:sp_prices_lower_bound1}, and get:
\begin{align*}
\sum_{c \in X} \pi_c  &= \sum_{c \in X \setminus W} \pi_c + \pi_z \geq  (\beta - z) \frac{|S| - \pi_z}{\beta \cdot \frac{n- |S|}{n} - z} + \pi_z \\
&= \frac{\beta |S| - |S| z - \beta \pi_z + z \pi_z + \beta \pi_z - \beta \pi_z \cdot \frac{|S|}{n} - z \pi_z}{\beta \cdot \frac{n- |S|}{n} - z}  \\
&= \frac{\beta |S| - |S| z - \beta \pi_z \cdot \frac{|S|}{n}}{\beta \cdot \frac{n- |S|}{n} - z}  > \frac{\beta |S| - |S| z - z \cdot \frac{n|S|}{n - |S|} \cdot \frac{|S|}{n}}{\beta \cdot \frac{n- |S|}{n} - z} \\
&= |S| \cdot \frac{\beta - z - \frac{z |S|}{n - |S|}}{\beta \cdot \frac{n- |S|}{n} - z} = |S| \cdot \frac{\beta n - \beta |S| -zn + z|S| - z |S|}{n - |S|} \cdot \frac{n}{\beta n - \beta |S| - zn} \\
&= \frac{n|S|}{n - |S|}\text{.}
\end{align*}
This again gives a contradiction with \eqref{eq:sp_prices_upper_bound}, and completes the proof.
\end{proof}

\section{Restrained EJR and Pareto Optimality}\label{appx:restrained-ejr}

In this section we provide the proofs of theorems that concern satisfiability of Restrained EJR. We highlighted and discussed these results in \Cref{sec:base_ejr_vs_restrained_ejr}, and below we provide all the technical details.

\newtheorem*{thmrestrainedejrexists}{\Cref{thm:restrained-ejr-exists}}

\begin{thmrestrainedejrexists}
For each election there exists an outcome satisfying Restrained EJR.
\end{thmrestrainedejrexists}
\begin{proof}
Consider an election $E = (C, N, \feasibles, \mathcal{A})$. Let $k$ be the maximum size of feasible outcomes, $k = \max_{W \in \feasibles}|W|$, and let $n$ denote the original number of voters---in the course of the proof we will be removing some voters yet the value of $n$ will not change.

Let us construct a feasible outcome through the following greedy procedure. We look for the largest integer value $\ell \geq 0$ such that there exists a group of not-yet removed voters $S$ with $|S| \geq \ell \cdot \nicefrac{n}{k}$ who approves a feasible set $T \subseteq \bigcap_{i \in S} A_i$, $|T| = \ell$. We select the candidates from $T$, remove the voters from $S$, and construct a new family of feasible sets:
\begin{align*}
\feasibles' = \left\{ W \subseteq C  \colon W \cup T \in \feasibles \right\} \text{.}
\end{align*}
We repeat the procedure recursively, for the election $E' = (C, N \setminus S, \feasibles', \mathcal{A})$. We finish, when we removed all the voters. Note that since after each update we have that $\emptyset \in \feasibles$, all the voters will be removed at some point. 

Let us call the final outcome $W$. We will show that $W$ satisfies Restrained EJR. Towards a contradiction assume that this is not the case, and let $S$ be a blocking coalition. Let $y = \max_{i \in S} u_i(W) + 1$ and let $k' = \left\lfloor \frac{|S|}{n} k\right\rfloor$. In particular, we have that: 
\begin{align*}
 k' \leq \frac{|S|}{n} k  \qquad \text{hence} \qquad   |S| \geq \frac{k' n}{k} \text{.}
\end{align*}

By applying the definition of Restrained EJR for $\hat{W} = \emptyset$ we infer that $y \leq k'$. Consider the first round $r$ when some voter from $S$ was removed. Note that until round $r$ we selected at most the following number of candidates:
\begin{align*}
(n - |S|) \cdot \nicefrac{k}{n} \leq  \left(n - \frac{k'n}{k} \right) \cdot \nicefrac{k}{n} = k - k' \text{.} 
\end{align*}
Now we define $\hat{W}$ as a subset of $k - k'$ candidates from $W$ containing all candidates selected before round $r$. Since $S$ is a blocking coalition we get that there exists $W'$ with $|W'| \leq k'$ such that:
$T = \hat{W} \cup W' \in \feasibles$, and  $\left|\bigcap_{i\in S} A_i \cap T\right| \geq y$. But this means that $T$ must be also feasible in round $r$.
Moreover, $S \geq \frac{k' n}{k} \geq y \cdot \frac{n}{k}$. Thus, in round $r$ the value of $\ell$ is at least equal to $y$. 
Thus, the removed voter must had approved at least $y$ candidates in the set selected in round $r$, which contradicts the fact that $y = \max_{i \in S} u_i(W) + 1$. 
This completes the proof.
\end{proof}

\begin{definition}[Pareto Optimality]
  We say that a feasible outcome $W$ is \emph{Pareto dominated} by a feasible outcome $W'$ if for each voter $i \in N$ it holds that $u_i(W') \geq u_i(W)$ and for at least one voter $i \in N$ the inequality is strict, i.e, $u_i(W') > u_i(W)$. An outcome $W$ is \myemph{Pareto optimal} if it is not Pareto dominated by any feasible outcome. 
\end{definition}

\newtheorem*{thmrestrainedejrandpareto}{\Cref{thm:restrained-ejr-and-pareto}}

Let $u$ be a \emph{utility vector} that specifies for each voter $i \in N$ the number of candidates $u(i)$ the voter approves in the selected outcome. 
We say that $W\subseteq C$ \emph{induces} the utility vector $u$ if for each $i\in N$ we have $|A_i\cap W|=u(i)$. 
Given a utility vector $u$ we define the \emph{total utility} among $S\subseteq N$ as $u_{\mathrm{tot}}(S)=\sum_{i \in S}u(i)$ and the \emph{average utility} as $u_{\mathrm{av}} = \frac{u_{\mathrm{tot}}(N)}{n}$. We may shortcut $u_{\mathrm{tot}}=u_{\mathrm{tot}}(N)$.

Before giving a formal proof, we provide a short sketch.
We take all possible utility vectors that have average utility $p$---we call such utility vectors \emph{normal}.
We associate with each such a utility vector $u$ a pairwise disjoint feasible committee $W_u$.
In the next step, for each $W_u$, we design many feasible committees that together witness that $W_u$ violates EJR. Each such violation of EJR  is certified by a set of voters $S\subseteq N$.
Next, we need to prove that the outcomes towards which $S$ deviated also either fail EJR or are Pareto dominated. There are two possible cases depending on the size of $S$.
Either $|S|$ is small, in which case we prove that such an outcome is Pareto dominated by a solution induced by some normal utility vector.
Otherwise, the total utility of a corresponding committee might be higher compared to the total utility of a normal vector.
In that case, we perform another transformation that creates a new utility vector that dominates the former one by a slight increase in the utility of a single voter only.
However, this allows us to design a new feasible committee using new candidates, but where we limit the total utility of the voters from $N\setminus S$.
We repeat both transformations again on a new carefully chosen set $S'\subseteq N\setminus S$, even more decreasing the total utility of the voters from $N\setminus (S\cup S')$.
This is very helpful, as we can eventually show that there must be many voters with relatively small utility in $N\setminus (S\cup S')$.
Finally, we perform a few more transformations on those voters, which decrease the total utility of the committee below the threshold set by the total utility of a normal utility vector.
Hence, we conclude the proof as, in each case, we were able to find sequences of EJR violations and Pareto dominations (these sequences can be viewed as a tree, where utility vectors correspond to nodes, and the transformations to the edges), where the last committee in each sequence (a leaf in a tree) is again Pareto dominated by a committee associated with a normal vector.

\begin{thmrestrainedejrandpareto}
There exists no selection rule that satisfies Restrained EJR and Pareto Optimality.
\end{thmrestrainedejrandpareto}
\begin{proof}
Let us fix $p = 100$ ($p$ can be any sufficiently large number). We fix the number of voters $n$, and the size of the committee $k$:
\begin{align*}
n = k = \frac{p^2}{2} + \frac{5p}{2} \text{.}
\end{align*}

First, we describe a collection of utility vectors $\mathcal{N}_p$, denoted as \emph{normal}.
The collection $\mathcal{N}_p$ consists of all possible utility vectors such that $u(i) \leq k$ for all $i$ and such that $u_{\mathrm{av}}=p$. In other words, the total utility is $u_{\mathrm{tot}}=np=kp$. 
For each $u\in \mathcal{N}_p$ we are going to create one feasible committee $W_u$ that induces $u$ such that for each $u,v \in \mathcal{N}_p$, $W_u$ and $W_v$ are disjoint.

\paragraph{Construction of $W_u$.}\phantomsection\label{constr}
Given a utility vector $u$ and $\succ$ an order on $N$, we construct $W_u\in \mathcal{F}$ as follows.
We often say that voters are \emph{consecutive} if they are consecutive with respect to order ${\succ}$.
We introduce a set $W_u\subseteq C$ of $k$ new candidates that were not part of any feasible committee constructed until this point.
We assign the candidates to the voters via round-robin---we put an arbitrary cyclic order on $W_u$ (the first candidate in $W_u$ follows after the last one).
We process the voters according to the order $\succ$.
At each step, we add one candidate to the approval set of voter $i$ at hand and move to the next candidate according to the cyclic order.
When the voter $i$ already approves $u(i)$ candidates, we move to the next voter according to $\succ$. 

This construction ensures an important property that will be used later on:
\begin{description}
  \item[Property A:] For a group $S$ of consecutive voters (wrt.\ $\succ$) such that $u_{\mathrm{tot}}(S)=x$, 
    each candidate is approved by at most $\lceil \nicefrac{x}{k} \rceil$ and at least $\lfloor \nicefrac{x}{k} \rfloor$ voters from $S$. \phantomsection \label{prop:A}
\end{description}

Also, note that each candidate is approved roughly the same number of times. In particular, for a normal utility vector this means that each candidate is approved by exactly $p$ voters.

Using the \hyperref[constr]{consturction}, we create $W_u$ for each $u\in\mathcal{N}_p$ where $\succ$ is defined as an order on $u(i)$ (where ties are broken arbitrarily). In other words, we process the voters starting from those with the highest expected utility until the ones with the lowest value of $u(i)$.

For each utility vector $u\in\mathcal{N}_p$ (and associated $W_u\in\mathcal{F}$), we construct a set of feasible committees that ensure that for $W_u$ there is a group of voters that deviates in $W_u$ and therefore EJR is violated.
Assume that there exists a group of $x>0$ voters $S$ such that each of them has the utility lower than $x$.
For any utility vector $u$ and $S\subseteq N$, we describe an $x$-\emph{transformation} as follows.
We add $x$ brand new candidates that we set to be jointly approved by $S$. 
This set of $x$ candidates, call it $X$, together with any subset of $k - x$ candidates from $W_u$ are added to $\mathcal{F}$.
Thus, $W_u$ does not satisfy EJR, since the group $S$ deviates.
Indeed, for any $T\subseteq W_u$ of $|T|\le k-x$ there is a feasible committee $W_u^T=T\cup X$, which gives each voter $i \in S$ strictly higher utility than $u(i)$.

Having $u$, $\succ$, $W_u$, a consecutive set $S\subseteq N$, and $T\subseteq W_u$ (where $|T|=k-x$), as defined above, we describe another transformation applied on $W_u^T$. 
We call it a \emph{PD-transformation}.
Let $u_T$ be the utility vector that $W_u^T$ induces.
We create a new utility vector $u^+$ by increasing the utility of a least satisfied voter $v\notin S$ by one.
Then we create a feasible committee $W_{u^+}$ on brand new candidates using the \hyperref[constr]{construction}, where $\succ'$ is an order in which we first put the voters from $S$, followed by the voters $i \in N \setminus S$ ordered according to $u^+(i)$ (where ties are broken arbitrarily).
This transformation clearly ensures that $W_u^T$ is Pareto dominated by $W_{u^+}$.

Now, we split the proof into two cases.
In the first case, $u\in\mathcal{N}_p$ is such that there is a group $S$ of $x$ voters, for $x < p$, such that each voter from $S$ approves less than $x$ candidates.
We perform $x$-transformation of $u$, and we obtain several new feasible sets $W_u^T$ for all $T\subset W_u$ where $|T|=k-x$.
However, each such $W_u^T$ (and a utility vector $u^T$ it induces) is Pareto dominated by some normal utility vector:
\begin{align*}
u^T_{\mathrm{tot}} = np - \underbrace{x \cdot p}_{\text{removing $x$ candidates }} + \underbrace{x \cdot x}_{\text{adding $x$ candidates}} <np \text{.} 
\end{align*}

Now, consider the other case, that is, when there is no such group $S$.
We claim that there must be a group of $(p+1)$ consecutive voters such that each of them approves $p-1$ or $p$ candidates.
Indeed, if this were not the case, then the total utility would be at least:
\begin{align*}
 u_{\mathrm{tot}} \ge &\underbrace{1 + 2 + \cdots + (p-1)}_{\text{no group deviating for $x < p$}} + \underbrace{(p-1)\cdot (p-1)}_{\text{at most $p$ voters with the utility of $p-1$}} + (n - 2p+2)\cdot (p + 1) \\
&\qquad=  \frac{p(p-1)}{2} + p^2 - 2p + 1 + np + n - 2p^2 -2p + 2p + 2 \\
&\qquad = \frac{p(p-1)}{2} + np + n  - p^2 -2p + 3 \\
&\qquad = np + n  - \frac{p^2}{2} -\frac{5p}{2} + 3 > np\text{,}
\end{align*}
a contradiction.

We let $S$ be a set of $p+1$ consecutive voters such that each of them approves $p-1$ or $p$ candidates. We do a $(p+1)$-transformation of $S$, and for each resulting committee we do a PD-transformation.
Therefore, we obtain utility vector $u'$ that has the following structure:

\begin{enumerate}
\item The total utility equals:
\begin{align}
  u'_{\mathrm{tot}}= np - \underbrace{(p+1)p}_{\text{removing $p+1$ candidates}} + \underbrace{(p+1)(p+1)}_{\text{adding $p+1$ candidates}} + \underbrace{1}_{\text{PD-transformation}} = np + p + 2 \text{.}\label{eq:tot}
\end{align}
\item From \hyperref[prop:A]{Property A} it follows that the total utility within the group $S$ is at least:
\begin{align}
  u'_{\mathrm{tot}}(S) &\geq  \underbrace{|S| (p-1)}_{\text{utility before transformation}} -  \underbrace{(p+1) \left \lceil \frac{|S| p} {k} \right\rceil}_{\text{\hyperref[prop:A]{Property A}}}  + \underbrace{|S|(p+1)}_{\text{candidates added through the transformation}}\nonumber \\
                                &= 2|S|p - (p+1) \left\lceil \frac{(p+1) p}{\frac{p^2}{2} + \frac{5p}{2}} \right\rceil \nonumber\\
                                &\geq 2(p+1)p - 2(p+1) = 2(p+1)(p-1) = 2p^2 - 2 \label{eq:Sutility}\text{.}
\end{align}
\item The total utility among $S$ is at most:
\begin{align}
  u'_{\mathrm{tot}}(S) &\leq  \underbrace{|S| p}_{\text{utility before transformation}}  + \underbrace{|S|(p+1)}_{\text{candidates added through the transformation}} \nonumber \\
                       &  =2p^2+3p+1 \label{eq:Sutilitylb}\text{.}
\end{align}
\end{enumerate}
This means that the voters from $N \setminus S$ have the total utility at most:
\begin{align*}
z'_{\mathrm{tot}} = np + p + 2 - 2p^2 + 2 = np - 2p^2 + p + 4 \text{.}
\end{align*}

To finish the proof, we need to make sure that each such resulting utility vector will be Pareto dominated by a normal utility vector after few transformations.

Consider such a utility vector denoted as $u'$. 
If there exists a group $S'$ of size $x$, $2 \leq x \leq p - 2$, where each voter approves fewer than $x$ candidates, then we do an $x$-transformation on $S'$. The total utility after such a transformation is at most:
\begin{align*}
   u'_{\mathrm{tot}}=np + p + 2 - xp + x^2  < np \text{.}
\end{align*}
For the last inequality, note that the expression $- xp + x^2$ takes its minimum at $x = \nicefrac{p}{2}$, and so it suffices to check the value in the extreme points, $x = 2$ and $x = p-2$. Thus, in this case, each utility vector being the result of such an $x$-transformation is Pareto dominated by some normal vector.

Hence, we can assume there is no such a group. Then we claim there must exist a group $S'$ of $p$ consecutive voters such that each voter from $S'$ approves $p-2$ or $p-1$ candidates and $S\cap S'=\emptyset$. Indeed, if this were not the case, then the total utility among $N \setminus S$ would be at least:
\begin{align*}
  u'_{\mathrm{tot}}(N\setminus S)\geq &\underbrace{0 +  2 + 3 + \cdots + (p-2)}_{\text{no group deviating for $x$}} + (p-2)(p-2) + (\underbrace{n - (p+1)}_{|N \setminus S|} - (p-2) - (p-2)) p \\
&\qquad= \frac{(p-2)(p-1)}{2} - 1 + p^2 - 4p + 4 + (n -3p +3)p \\
&\qquad= \frac{(p-2)(p-1)}{2} + np -2p^2 -p +3  > np - 2p^2 + p + 4 = z'_{\mathrm{tot}} \text{.}
\end{align*}
The last inequality holds for each sufficiently large $p$, so we get a contradiction.

Thus, we may assume that such $S'$ exists and we perform a $p$-transformation on it and then a PD-transformation on each outcome 
(there $\succ''$ order is defined as follows: we first put the voters from $S$, next we put the voters from $S'$, followed by the voters $i \in N\setminus (S\cup S')$ ordered by $u'(i)$). 
After such transformations, there will be $2p+1$ voters (those from $S$ and $S'$) that will have at least the following total utility.
Note that as $S$ is consecutive in $u'$, we can use \hyperref[prop:A]{Property A} on $S$.
Let $u''$ be the resulting utility vector.
\begin{align*}
  u''_{\mathrm{tot}}(S\cup S')
  &\geq  \underbrace{|S'| (p-2)}_{\text{utility before transformation}} -  \underbrace{p \left \lceil \frac{|S'| (p-1)} {k} \right\rceil}_{\text{\hyperref[prop:A]{Property A} on $S'$}}  + \underbrace{|S'|p}_{\text{candidates added}} \\
   & \qquad\qquad+  \underbrace{2p^2-2}_{\text{utility of $S$~(\ref{eq:Sutility})}} -\underbrace{p \left \lceil \frac{2p^2+3p+1} {k} \right\rceil}_{\text{\hyperref[prop:A]{Property A} on $S$ + (\ref{eq:Sutilitylb})}} \\
                                &\geq p^2-2p -2p +p^2 +2p^2-2- 4p \\
   &= 4p^2 - 8p -2 \text{.} 
\end{align*}

It is easy to check that after such a transformation, the total number of approvals will increase by at most one: 
%
\begin{align}
  u''_{\mathrm{tot}} &\leq \underbrace{ np+p+2}_{\text{utility of $u'$ (\ref{eq:tot})}} - \underbrace{p \left\lfloor \frac{np+p+2}{k} \right\rfloor}_{\text{removing $p$ candidates}} + \underbrace{p^2}_{\text{adding $p$ candidates}} + \underbrace{1}_{\text{PD-transformation}}\nonumber \\
                     &\leq np + p + 3 \text{.} \label{eq:tot2}
\end{align}

Consequently, the total number of approvals among the voters from $N \setminus (S \cup S')$ is at most:
 \begin{align*}
z''_{\mathrm{tot}} = np + p + 3 -  4p^2 + 8p  + 2 = np  - 4p^2 + 9p + 5 \text{.}
\end{align*}

Now, we repeat the reasoning once again on the resulting utility vectors.
If there exists a group $S^*$ of size $x$, $2 \leq x \leq p - 2$, where each voter approves fewer than $x$ candidates, then we do an $x$-transformation and obtain utility vectors that are Pareto dominated by normal vectors. Otherwise, we will show that there exist at least two disjoint groups $S_1, S_2 \subseteq N \setminus (S \cup S')$ of size $p-1$ and $p$, respectively, that have fewer than  $p-1$ representatives. Indeed, if this were not the case, then the total utility among $N \setminus (S \cup S')$ would need to be at least:
\begin{align*}
&\underbrace{0 +  2 + 3+ \cdots + (p-2)}_{\text{no group deviating for $x$}} + p(p-2) + (\underbrace{n - (p+1) - p}_{| N \setminus (S \cup S')|} - (p-2) - p) (p-1) \\
&\qquad= \frac{(p-2)(p-1)}{2} - 1 + p^2 -2p + 2 + (n - 4p + 1)(p-1) \\
&\qquad= \frac{(p-2)(p-1)}{2} + np -3p^2 -p + 1 - n + 4p - 1 \\
&\qquad= \frac{(p-2)(p-1)}{2} + np -3p^2 +3p  -\frac{p^2}{2} - \frac{5p}{2} \\
&\qquad > np  - 4p^2 + 9p + 5 =z''_{\mathrm{tot}}\text{.}
\end{align*}
The last inequality holds for sufficiently large values of $p$ (one can compare only the coefficients for $p^2$). We get a contradiction, and so we have the two aforementioned groups $S_1$ and $S_2$.
We perform $(p-1)$-transformation on $S_1$ and, on the results, a PD-transformation (with $\succ$ set arbitrarily).
Either the total utility after such transformations is already lower than $np$ (and so it is Pareto dominated by a normal vector) or each candidate appears in at least $p$ approval sets. 
Moreover, observe that the total utility of $S_2$ only decreased (except possibly for one voter $v$).
Hence, we perform the last $(p-1)$-transformation on $S_2\setminus\{v\}$.
Therefore, the total utility of a utility vector $u^*$ induced by a committee created after all such transformations is at most: 
\begin{align*}
  u^*_{\mathrm{tot}} &\leq \underbrace{np + p + 3}_{u''_{\mathrm{tot}} (\ref{eq:tot2}) } - \underbrace{2(p-1)p}_{\ge \text{decrease for both $(p-1)$-transf.}} +\underbrace{2(p-1)(p-1)}_{\text{$\le$ increase for both $(p-1)$-transf.}} + \underbrace{1}_{\text{PD-transformation}}  \\
& = np + p + 4 -2p^2 + 2p + 2p^2 - 4p + 2 = np -p + 6 < np \text{.}
\end{align*}
Thus, at the end, all the utility vectors are Pareto dominated by some normal vectors. This completes the proof.
\end{proof}

\section{Conclusion}

We have considered a general model of social choice, where the structure of output is given through feasibility constraints. The feasibility constraints allow to encode different types of elections (e.g., single-winner and multi-winner elections, participatory budgeting, judgment aggregation, etc.). We have proposed a new technique of extending classic notions of proportionality to the general model of social choice with feasibility constraints. This way we have defined the axioms of justified representation in the model with constraints. Our technique also allows to extend other notions of fairness such as the core (see \Cref{sec:core}).

Our strongest notion of proportionality, fully justified representation, is always satisfiable, even for general monotone utility functions. We further show that natural adaptations of two committee election rules, Proportional Approval Voting and Phragm\'en's Sequential Method, satisfy strong notions of proportionality if and only if the feasibility constrains are matroids. Phragm\'en's Sequential Method additionally provides good approximation of some of our notions of fairness for non-matroid constraints. This makes the rule suitable for elections of different type and structure. We have also generalised the concept of stable-priceabiliy to the case of elections with constraints.

There are several pressing open questions. First, our work mainly focuses on approval ballots and corresponding dichotomous utility functions; many applications, however, require dealing with more generic utility functions. Specifically, we are interested in the following two questions:
\begin{inparaenum}[(1)]
\item Can we meaningfully define Phragm\'en's Sequential Method for additive utility functions, so that the rule preserves its most compelling properties?
\item Can we define Method of Equal Shares for elections with constraints? The answer to the second question seems challenging. The main difficulty lies in the fact that we do not know how to set the prices of the candidates. If they are set too high, then some groups of voters might not be able to afford to buy enough supported candidates. If we set them too low, then the groups might be left with money which cannot be used for buying candidates without breaking feasibility constrains.
\end{inparaenum}

It is further important to check how the considered rules perform on real and synthetic data.

The setting with weighted candidates also remains mostly unexplored. This setting is particularly important, since it models the increasingly popular process of participatory budgeting. Can we define an analogue of matroid constraints for weighted candidates? This seems plausible given that the exchange property seems to be naturally adaptable to weights. Can we design rules that satisfy the strong proportionality axioms for such constraints?

\subsection*{Acknowledgments}

We thank Yiheng Shen, Kamesh Munagala, Nikhil Chandak, Dominik Peters for the very helpful discussions and their invaluable remarks. 

Tomáš Masařík was supported by Polish National Science Centre SONATA-17 grant number 2021/43/D/ST6/03312. Grzegorz Pierczyński was supported by by Polish National Science Centre PRELUDIUM grant number UMO-2022/45/N/ST6/00271. Piotr Skowron is supported by the European Union (ERC, PRO-DEMOCRATIC, 101076570). Views and opinions expressed are however those of the author(s) only and do not necessarily reflect those of the European Union or the European Research Council. Neither the European Union nor the granting authority can be held responsible for them.
\begin{center}
  \includegraphics[width=3cm]{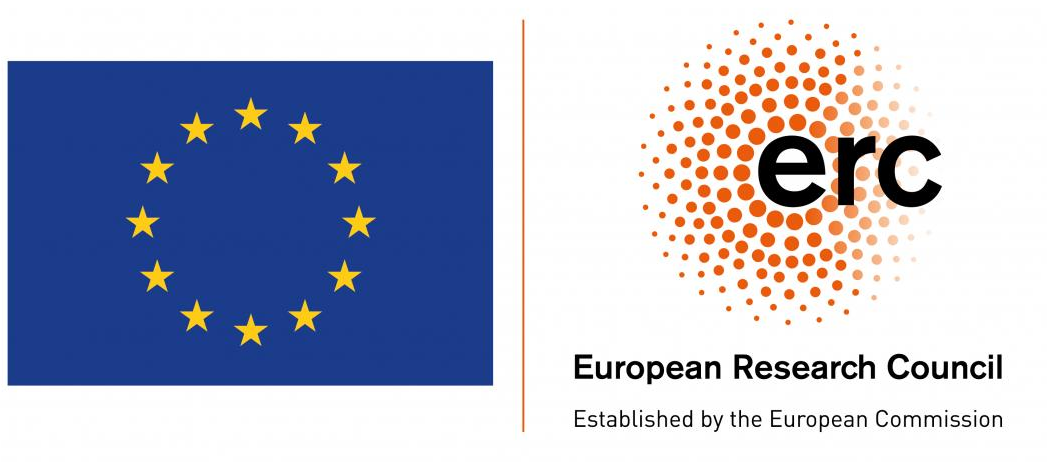}
\end{center}

\bibliographystyle{plainnat}
\bibliography{fairness}

\appendix

\section{Core in the General Model with Constraints}\label{sec:core}
Our approach also allows to extend other concepts of fairness to the model with constraints. Here we explain how to extend the concept of core.
We propose both variants of core to study. First, we define base variant of the core.

\begin{definition}[Base Core]\label{def:Basecore}
  Given an election $E = (C, N, \feasibles, \mathcal{A})$ we say that a group of voters $S \subseteq N$ is $(\alpha, \beta)$-\emph{cohesive}, $\alpha \in \naturals$, $\beta\colon S \to \reals^{+}$, if for each feasible set $T \in \feasibles$ either there exists $X \subseteq C$ such that $|X| = \alpha$, $T \cup X \in \feasibles$, and $u_i(X) \geq \beta(i)$ for each $i \in S$, or
\begin{align*}
\frac{|S|}{n} > \frac{\alpha}{|T| + \alpha} \text{.}
\end{align*}
We say that a feasible outcome $W \in \feasibles$ of an election $E = (C, N, \feasibles, \mathcal{A})$ is in the \myemph{base core} if for each $\alpha \in \naturals$, $\beta\colon S \to \reals^{+}$, and each  $(\alpha, \beta)$-cohesive group of voters $S\subseteq{}N$ there exists a voter $i \in S$ such that $u_i(W) \geq \beta(i)$.  \hfill $\lrcorner$
\end{definition}

Analogously we define the concept of core.

\begin{definition}[Core]\label{def:core}
  Given an election $E = (C, N, \feasibles, \mathcal{A})$ and a feasible outcome $W$ we say that a group of voters $S \subseteq N$ is $(\alpha, \beta)$-\emph{cohesive} in $W$, $\alpha \in \naturals$, $\beta\colon S \to \reals^{+}$, if for each subset $T \subseteq W$ either there exists $X \subseteq C$ such that $|X| = \alpha$, $T \cup X \in \feasibles$, and $u_i(X) \geq \beta(i)$ for each $i \in S$, or
\begin{align*}
 \frac{|S|}{n} > \frac{\alpha}{|T| + \alpha} \text{.}
\end{align*}
We say that a feasible outcome $W \in \feasibles$ of an election $E = (C, N, \feasibles, \mathcal{A})$ is in the \myemph{core} if for each $\alpha \in \naturals$, $\beta\colon S \to \reals^{+}$, and each group of voters $S\subseteq{}N$ that is  $(\alpha, \beta)$-cohesive in $W$ there exists a voter $i \in S$ such that $u_i(W) \geq \beta(i)$.  \hfill $\lrcorner$
\end{definition}

Our definition of the core clearly implies the definition of FJR and corresponds to the definition of the core for committee election rules.

\section{Computational Social Choice Models and Matroids}\label{sec:matroid}
In this section we prove that the examples of matroid feasibility constraints provided in \Cref{sec:feasibility_constraints} satisfy \ref{cond:exchange-property}, and the examples of non-matroid feasibility constraints violate it.

\begin{description}
    \item[Committee elections.] Consider two feasible sets $X, Y$ such that $|X| < |Y|$. It is clear that $|Y| \leq k$ and for each $c\in Y \setminus X$ we have that $|X\cup\{c\}| \leq k$, hence by the definition of $\feasibles$,  $X\cup\{c\}\in \feasibles$, which shows that $\feasibles$ satisfies condition \ref{cond:exchange-property}.
    \item[Public decisions.] Consider two feasible sets $X, Y$ such that $|X| < |Y|$. Then for at least one binary issue $C_r$ ($r\in [z])$, we have that $X\cap C_r = \emptyset$ and $Y\cap C_r \neq \emptyset$ (hence, $|Y\cap C_r| = 1$). Then after adding the candidate from $Y\cap C_r$ to $X$, $X$ is still feasible, which shows that $\feasibles$ satisfies condition \ref{cond:exchange-property}.
    \item[Committee elections with disjoint attributes.]  Consider two feasible sets $X, Y$ such that $|X| < |Y|$. If for some attribute $r$ we have that $|X\cap C_r| < \lquota{r}$ and $(Y\setminus X)\cap C_r \neq \emptyset$, then after adding any candidate from the latter set to $X$, $X$ is still feasible. Suppose now that it is not the case, i.e., $(*)$ $(Y\setminus X)\cap C_r = \emptyset$ for each $r$ such that $|X\cap C_r| < \lquota{r}$. Naturally, from the construction of $\feasibles$, we have also that $|X\cap C_r| = |Y\cap C_r|$ for each $r$ such that $|X\cap C_r| = \uquota{r}$. Therefore, since $|X| < |Y|$, there need to exist attribute $r$ such that $\lquota{r} \leq |X \cap C_r| < |Y \cap C_r| < \uquota{r}$. Consider now any candidate $c\in (Y\setminus X) \cap C_r$. The set $X\cup\{c\}$ does not violate any upper quotas and is still possible to be completed to a $k$-sized set so that all lower quotas are satisfied (because $Y$ is possible to be completed in such a way, $X\cup \{c\}$ has at most the same size as $Y$, and from $(*)$ the number of seats required to satisfy all the lower quotas is no greater in $X\cup\{c\}$ than in $Y$), hence it is feasible and \ref{cond:exchange-property} is satisfied.
    \item[Ranking candidates.] Suppose that we need to elect a ranking among $3$ candidates $c_1, c_2, c_3$. Consider set $X=\{c_{1, 2}, c_{2, 3}\}$ and set $Y=\{c_{3, 2}, c_{2, 1}, c_{3, 1}\}$. Then it is clear that no candidate from $Y$ can be added to $X$ so that $X$ still represents a valid ranking, hence \ref{cond:exchange-property} is violated.
    \item[Committee elections with negative votes.] Consider set $X$ containing some $k$ real candidates $c_1, c_2, \ldots, c_k$ and a $(k+1)$-sized set $Y=\{\bar{c}_1, c_2, c_3, \ldots, c_k, c_{k+1}\}$. Now $Y\setminus X = \{\bar{c}_1, c_{k+1}\}$. None of them can be added to $X$ without breaking feasibility constraints---adding $\bar{c}_1$ would mean that $c_1$ is both elected and unelected, and adding $c_{k+1}$ would mean that more than $k$ candidates are elected. Hence, \ref{cond:exchange-property} is violated.
    \item[Judgement aggregation.] Consider two variables $x$ and $y$. As described in \Cref{sec:feasibility_constraints}, we introduce four candidates $c_{x, T}$, $c_{x, F}$, $c_{y, T}$, $c_{y, F}$. Now suppose that we require that the following formula holds: $x \implies \neg y$. Consider set $X=\{c_{x, T}\}$ and set $Y=\{c_{x, F}, c_{y, T}\}$. Then it is clear that no candidate from $Y$ can be added to $X$ so that $X$ is still feasible, hence \ref{cond:exchange-property} is violated.
\end{description}

\end{document}